%% file: main.tex
\documentclass[sigconf,nonacm]{acmart}

\setcopyright{none}
\usepackage{stfloats}
\usepackage{algorithm}
\usepackage{algpseudocode}
\usepackage{multirow}
\usepackage{adjustbox}
\usepackage{xspace}
\usepackage{amsthm}

\newcommand{\nlist}{n_\mathrm{list}}
\newcommand{\nprobe}{n_\mathrm{probe}}
\newcommand{\stp}[1]{\textcircled{\raisebox{-0.5pt}{\scriptsize #1}}}
\newcommand{\OURS}{RaBitQ-SSD\xspace}

\begin{document}
\sloppy
\title{RaBitQ-SSD: Split Codes and Pipelined I/O for SSD-Resident Vector Search}

\settopmatter{authorsperrow=4}

\author{Yuexuan Xu}
\affiliation{%
  \institution{Nanyang Technological University, Singapore}
 \city{}
  \country{}
}
\email{yuexuan001@e.ntu.edu.sg}

\author{Jianyang Gao}
\affiliation{%
  \institution{ETH Zurich}
  \city{Zurich}
  \country{Switzerland}
}
\email{jianyang.gao@inf.ethz.ch}

\author{Michael Norris}
\affiliation{%
  \institution{Meta}
  \city{Menlo Park}
  \country{United States}
}
\email{mnorris@meta.com}

\author{Alibek Zhakubayev}
\affiliation{%
  \institution{Meta}
  \city{Menlo Park}
  \country{United States}
}
\email{alibek@meta.com}

\author{Pankaj Singh}
\affiliation{%
  \institution{Meta}
  \city{Bellevue}
  \country{United States}
}
\email{pankajsingh@meta.com}

\author{Junjie Qi}
\affiliation{%
  \institution{Meta}
  \city{Bellevue}
  \country{United States}
}
\email{jqi@meta.com}

\author{Matthijs Douze}
\affiliation{%
  \institution{Meta}
  \city{Paris}
  \country{France}
}
\email{matthijs@meta.com}

\author{Cheng Long}
\affiliation{%
  \institution{Nanyang Technological University, Singapore}
 \city{}
  \country{}
}
\email{c.long@ntu.edu.sg}

\renewcommand{\shortauthors}{Xu et al.}

\begin{abstract}

SSD-resident approximate nearest-neighbor search is essential when vector collections exceed DRAM capacity.
The challenge is to reduce SSD reads and hide I/O latency
through concurrent reads and overlap with computation.
However, for graph-based search, progressive candidate discovery limits advance I/O planning while IVF search reads inverted lists in full, incurring unnecessary SSD reads.
In this paper, we present \OURS, an extension of IVF-RaBitQ for SSD-resident vector search.
Specifically, we propose Split-RaBitQ, which retains a configurable
prefix of each binary code in DRAM, allowing in-memory code storage
below one bit per dimension.
Using this partial representation, it provides an unbiased distance
estimator with a probabilistic error bound.
Once the in-memory coarse quantizer identifies candidate lists,
this bound supports pruning individual candidates within them,
enabling finer-grained SSD access.
We also design an asynchronous search pipeline that coordinates candidate
pruning with SSD read scheduling to reduce unnecessary reads while
overlapping I/O with computation.
On datasets ranging from $5$ million to $1$ billion vectors, \OURS
delivers up to $1.74\times$ the throughput of graph-based baselines at $90\%$ recall,
while reducing SSD page reads by up to $3.8\times$.
Its on-SSD indexes are up to $7.0\times$ smaller and $10.0\times$
faster to build than DiskANN.
We also build and search an index of $10$ billion vectors on a single
machine with one SSD, achieving over $3{,}000$ queries per
second at $90\%$ recall.

\end{abstract}

\maketitle

\pagestyle{plain}
\thispagestyle{plain}

\input{content/1-introduction}
\input{content/2-preliminaries}
\input{content/3-Split-RaBitQ}
\input{content/4-RaBitQ-SSD-pipeline}
\input{content/5-experiment}
\input{content/6-related_work}
\input{content/7-conclusion}

\bibliographystyle{ACM-Reference-Format}
\bibliography{main}


\input{content/appendix_proof}

\end{document}

%% file: content/1-introduction.tex
\section{Introduction}
\label{ref:intro}

\begin{figure}[!t]
  \centering
  \includegraphics[width=\linewidth]{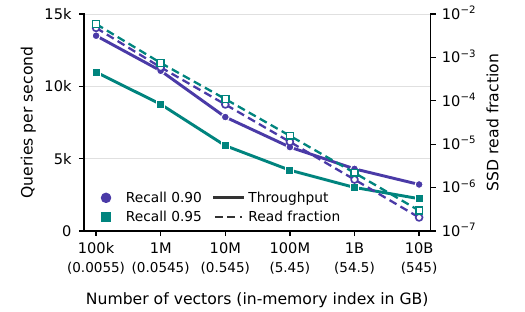}
  \caption{\OURS scales from 100 thousand to 10 billion
  DINO vectors ($1,024$ dimensions) on one server with one SSD.
  Solid lines show throughput for a given recall target (left axis); dashed lines show the fraction
  of the on-SSD data that is read to perform the search (right axis, log scale): at 10B scale it is less than one millionth.
  }
  \vspace{-3mm}
  \label{fig:dino-scaling-splash}
\end{figure}

Modern embedding models represent unstructured data, such as images and text, as high-dimensional vectors, typically with hundreds to thousands of dimensions.
Approximate nearest-neighbor (ANN) search retrieves vectors close to a query without exhaustively scanning the collection and underlies applications such as semantic search, recommendation, and retrieval-augmented generation~\cite{lewis2020rag,openai2024embedding}.
As collections grow, however, the vectors themselves can exceed the memory capacity of a single machine: one billion $1024$-dimensional vectors in FP32 already occupy $4$\,TB.
Flash storage provides a natural storage tier for such collections: modern flash-based SSDs offer roughly $10$--$50\times$ lower cost per byte than DRAM
and a single PCIe~4.0 NVMe SSD can deliver over one million random-read IOPS~\cite{haas2023nvme}.
Consequently, storing the vectors on SSD while keeping compressed representations in memory has become a practical and cost-effective design for billion-scale ANN search~\cite{subramanya2019diskann,chen2021spann,wang2024starling,guo2025pipeann}.

DiskANN~\cite{subramanya2019diskann} established a widely adopted graph-based design
for SSD-based ANN search in which compact quantized codes in DRAM guide
the search while the graph and full-precision vectors reside on SSD.
Its synchronous beam search batches  reads to exploit device
parallelism, but each search round waits for those reads to complete
before the next round, limiting I/O--computation overlap.
Starling~\cite{wang2024starling} reduces SSD I/O through locality-aware graph reordering and shorter traversal paths.
However, as vector dimensionality grows and vertex records become larger,
fewer related vertices can be co-located within the same SSD block.
PipeANN~\cite{guo2025pipeann} breaks the strict compute--I/O ordering of conventional beam search
by speculatively issuing reads while computation proceeds.
However, later I/O completions may reveal better candidates, making some speculative reads unnecessary.
Subsequent work further optimizes disk layouts and
caching~\cite{ni2023diskannpp,kang2025pageann,bamg,govector},
the placement of compressed
codes~\cite{tatsuno2024aisaq,lm_diskann,zhao2026veloann},
distance computation~\cite{chen2026alayalaser},
and I/O--computation scheduling~\cite{kang2026laann}.
Despite these optimizations, graph methods with SSD-resident adjacency lists retain a structural constraint: candidates are discovered
progressively as the traversal proceeds.
Expanding an uncached vertex requires accessing its on-disk record to obtain its neighbors
for subsequent exploration.
This causality
limits how far SSD accesses can be planned ahead of time.

On the other hand, IVF-based indexes organize vectors into inverted lists,
avoiding the cost of building and storing a proximity graph over the full
collection. An in-memory coarse index identifies candidate lists before
accessing their SSD-resident contents.
Among these methods, SPANN~\cite{chen2021spann} uses query-aware pruning to select inverted lists,
while HELMSMAN~\cite{huang2026helmsman} learns how many lists to probe.
Both read the selected lists in full from SSD.

This motivates finer-grained I/O pruning within probed inverted lists.
If an in-memory representation provides distance estimates with lower bounds for each vector,
the bounds can help rule out vectors that cannot improve the current $k$NN result before fetching their SSD-resident data.
The in-memory IVF-RaBitQ~\cite{RaBitQ,ExtendedRaBitQ,rabitqlibrary} supports such
error-bounded pruning by maintaining two representations of each vector.
It scans a compact $D$-bit RaBitQ code to
estimate distances and prune candidates, and accesses a higher-precision
representation only for the surviving candidates,
where $D$ denotes the vector dimensionality.
This separation naturally suggests an SSD-based extension: retain the compact
codes in DRAM for pruning and place the higher-precision representations on
SSD for re-ranking, but this exposes two challenges.
First,
this direct extension retains the full binary code
and cannot accommodate memory budgets below one bit per dimension that is necessary on billion scale datasets,
e.g., for 10 billion 1024-dimensional vectors, the binary codes alone require about 1.28TB of DRAM.
Second, the in-memory pruning needs careful coordination with I/O pipelining so as to reduce unnecessary reads and enable I/O--computation overlap, e.g., fetching the candidates surviving the pruning individually wastes page bandwidth, issuing pages too early can read data that later pruning would eliminate, etc.

In this paper, we present \OURS, an SSD-based extension of
in-memory IVF-RaBitQ that addresses these two challenges.
First, we propose Split-RaBitQ (\S\ref{sec:split}), which keeps only a \emph{prefix} of $D_m \le D$ bits of each binary code in memory.
Therefore, the memory consumption is configurable by setting $D_m$.
We derive an
unbiased distance estimator and a corresponding error bound from the
prefix alone, enabling error-bounded pruning.
Second, we develop an asynchronous SSD search
pipeline that reduces SSD reads while overlapping I/O with
computation. The pipeline coalesces surviving candidates at page
granularity and issues page reads asynchronously under a bounded
in-flight I/O budget. It prioritizes more promising pages and
re-evaluates them against the latest distance threshold before issuance,
allowing updated search results to eliminate many unnecessary SSD
reads without sacrificing I/O concurrency.

In summary, we make the following contributions.
\begin{itemize}
\item We propose Split-RaBitQ (\S\ref{sec:split}), which keeps a
configurable prefix of each RaBitQ binary code in DRAM.
We derive an unbiased distance estimator and a
probabilistic error bound from the prefix, enabling candidate pruning
with an in-memory code footprint below one bit per dimension.

\item We design an asynchronous search pipeline
(\S\ref{sec:pipeline}) that reduces unnecessary SSD I/O while
sustaining I/O--computation overlap.

\item We conduct an extensive evaluation of \OURS
(\S\ref{sec:experiment}) on six datasets ranging from $5$M to $1$B
vectors and $512$--$1536$ dimensions, comparing against strong baselines.
We further evaluate scalability from $100$ thousand to
$10$ billion DINO vectors on a single server with one SSD
(Figure~\ref{fig:dino-scaling-splash}).

\end{itemize}

\OURS is open-sourced at \url{https://github.com/VectorDB-NTU/RaBitQ-SSD}.

%% file: content/2-preliminaries.tex
\section{Preliminaries}
\label{sec:pre}

\subsection{ANN Search and IVF}
\label{sec:pre:ann}

Let $\mathcal{X} \subset \mathbb{R}^D$ be a collection of $N$ data
vectors, and let $q_r \in \mathbb{R}^D$ be a query vector. Given a
distance function, a $k$-nearest-neighbor ($k$NN) query asks for the $k$
vectors in $\mathcal{X}$ closest to $q_r$. In large-scale vector
search, this query is typically answered approximately and thus is called approximate nearest neighbor (ANN) query: an ANN
algorithm returns a set $\mathcal{R}$ of $k$ vectors, whose accuracy is
measured by $\mathrm{recall}@k = |\mathcal{R} \cap \mathcal{G}| / k$,
where $\mathcal{G}$ denotes the exact $k$NN set.

The inverted file (IVF)~\cite{jegou2011pq} is a classical
partitioning-based ANN index that underlies many widely deployed vector
search systems~\cite{douze2026faisslib,wang2021milvus,wei2020analyticdb}.
At build time, the dataset is partitioned into $\nlist$ clusters, typically by $k$-means.
The centroids $c_1, \dots, c_{\nlist}$ serve as a coarse quantizer, each vector being assigned to its nearest centroid, and the vectors assigned to $c_i$ form the inverted list $L_i$.
At query time, IVF scans the $\nprobe$ lists whose centroids are closest to $q_r$, in ascending order of centroid distance; the parameter $\nprobe$ controls the accuracy--efficiency trade-off.
Though IVF narrows the search to the probed lists, scanning a list still requires estimating the distance from $q_r$ to every vector in it, so fast distance estimation with a small memory footprint is crucial.
It can be achieved through vector quantization such as RaBitQ~\cite{RaBitQ,ExtendedRaBitQ}, which we review next.

\subsection{RaBitQ and Its Error Bound}
\label{sec:pre:rabitq}

RaBitQ~\cite{RaBitQ} quantizes a vector relative to a center vector, for which the cluster centroid is the natural choice.
RaBitQ's notation is used throughout: for a raw data vector $o_r$ assigned to centroid $c$ and a raw query $q_r$,
$o := (o_r - c)/\lVert o_r - c\rVert$ and
$q := (q_r - c)/\lVert q_r - c\rVert$ are the unit residuals.
The distance metric is $d(q_r, o_r)=\lVert o_r - q_r\rVert^2$.
Expanding either around $c$ yields $d(q_r, o_r)=-2\,\lVert o_r - c\rVert\,\lVert q_r - c\rVert\,\langle o, q\rangle + K$.
$K$ groups the terms precomputed per vector or computed once per probed cluster~\cite{RaBitQ}.
Therefore the per-vector computation reduces to estimating the inner product between two unit vectors $\langle o, q\rangle$, and RaBitQ focuses on quantizing unit vectors.
This decomposition works similarly when the nearest neighbor search in L2 space is replaced by maximum inner product search.

The codebook of RaBitQ is the set of $2^D$ vectors $\{\pm 1/\sqrt{D}\}^D$ rotated
by a random orthogonal matrix $P$, sampled once per index. The
quantized vector of $o$ is the nearest codeword $\bar{o} := P\bar{x}$,
$\bar{x} \in \{\pm 1/\sqrt{D}\}^D$; since
$\langle o, P\bar{x}\rangle = \langle P^{-1}o, \bar{x}\rangle$,
$\bar{x}$ takes the
signs of the coordinates of the rotated vector $P^{-1}o$, stored as the
$D$-bit \emph{binary code} $\bar{x}_b \in \{0,1\}^D$ with
$\bar{x} = (2\bar{x}_b - \mathbf{1})/\sqrt{D}$.

RaBitQ estimates $\beta=\langle o, q\rangle \approx \widehat{\beta}_D =
\langle \bar{o}, q\rangle / \langle \bar{o}, o\rangle$,
where $\langle \bar{o}, o\rangle$ is a per-vector scalar precomputed at indexing time as follows~\cite{RaBitQ}.
\begin{equation}
\label{eq:alignment}
\langle \bar{o}, o\rangle
= \langle \bar{x}, P^{-1}o\rangle
= \frac{1}{\sqrt{D}} \sum_{i=1}^{D} \bigl\lvert (P^{-1}o)[i] \bigr\rvert
\end{equation}

\begin{lemma}[Theorem~3.2 of~\cite{RaBitQ}]
\label{lem:rabitq-bound}
$\widehat{\beta}_D = \langle \bar{o}, q\rangle / \langle \bar{o}, o\rangle$ is an unbiased
estimator of $\beta=\langle o, q\rangle$, and
\begin{equation*}
\Pr\Biggl\{\,
\Bigl\lvert
\beta -
\widehat{\beta}_D
\Bigr\rvert
>
\sqrt{\frac{1 - \langle \bar{o}, o\rangle^2}
           {\langle \bar{o}, o\rangle^2}}
\cdot \frac{\epsilon_0}{\sqrt{D-1}}
\,\Biggr\}
\;\le\; 2 e^{-c_0 \epsilon_0^2} ,
\end{equation*}
where $\epsilon_0$ is a parameter that controls the failure probability
and $c_0$ is a constant.
\end{lemma}

\noindent
Lemma~\ref{lem:rabitq-bound} gives a confidence interval for
$\langle o, q\rangle$: the estimate plus or minus the error bound above.
Since the coefficient of $\langle o, q\rangle$ in $d(q_r, o_r)$ is negative, substituting the upper endpoint of its confidence interval yields a lower bound $\mathrm{LB}(q_r, o_r)$ on the distance that holds with probability at least $1 - 2e^{-c_0\epsilon_0^2}$.
The in-memory IVF-RaBitQ uses this
bound to prune candidates before re-ranking (\S\ref{sec:pre:ivfrabitq}).

The numerator of the estimate $\langle \bar{o}, q\rangle
= \langle \bar{x}, P^{-1}q\rangle$ is
a sum of $D$ coordinate-wise products that can be stored in a lookup table built once per query~\cite{RaBitQ,rabitqlibrary}.
The vectors are laid out in batches of 32 packed codes so that the lookups can be performed from SIMD registers, a technique called \emph{FastScan}~\cite{andre2015fastscan}.
All other terms of the estimate and
the bound that depend on the data vector are folded at indexing time
into three per-vector scalars, the \emph{factors}, stored beside the
code~\cite{rabitqlibrary}.
For higher accuracy, Extended
RaBitQ~\cite{ExtendedRaBitQ} spends $B$ bits per coordinate: the most
significant bit is the sign bit above, and the remaining $B-1$ bits constitute the \emph{ex-code}; its estimate needs two further factors per
vector~\cite{ExtendedRaBitQ,rabitqlibrary}.
The multi-bit RaBitQ distance estimation~\cite{ExtendedRaBitQ} reuses the sign-bit sum and refines it with the ex-code part.

\subsection{In-Memory IVF-RaBitQ and a Naive SSD Extension}
\label{sec:pre:ivfrabitq}

The in-memory IVF-RaBitQ~\cite{RaBitQ,ExtendedRaBitQ,rabitqlibrary} integrates Extended RaBitQ into
an in-memory IVF index.
Each inverted list stores the binary codes of its vectors packed in
batches of 32 for FastScan.
The search maintains \emph{KNNSet}, the $k$ candidates
with the smallest refined distances found so far.
The threshold $\mathrm{dist}_k$ is the largest of these distances,
or $+\infty$ before $k$ candidates have been collected.
In a probed list, every batch yields estimates and lower bounds
(for batches of 32 candidates). Candidates with
$\mathrm{LB}(q_r, o_r) < \mathrm{dist}_k$ are re-ranked using their
ex-codes and may update KNNSet; all others are pruned.
No raw vectors are kept; the refined distances are
used for the final ranking~\cite{ExtendedRaBitQ}.

\label{sec:pre:naive}%
The most direct way to port this index to SSD, suggested
in~\cite{ExtendedRaBitQ}, keeps the binary code and the three factors in memory, and moves the ex-codes, with their two
re-ranking factors, to SSD.
The search is unchanged up to re-ranking:
when a candidate survives pruning, the page holding its ex-code is read
from SSD, scanning blocks until the read completes, and the candidate
is re-ranked. However, this naive extension has two problems.

\noindent\textbf{Problem 1: a rigid memory footprint.}
The binary code has a fixed size, one bit per coordinate, so the
in-memory part of the index occupies $ND/8$ bytes plus the factors,
which is over 1TB when $N$ is 10 billion and $D = 1024$.
In this paper, we present Split-RaBitQ (\S\ref{sec:split}),
which keeps only a \emph{prefix} of the binary code in memory,
and show that the prefix still yields a lower bound for pruning.
As a result, the size of the in-memory code is configurable by setting the length of the prefix.

\label{sec:pre:ssd}%
\noindent\textbf{Problem 2: blocking I/O.}
A random SSD read of $4$\,KB typically takes 10s of microseconds.
Issuing one read at a time therefore completes $\sim$10k
of reads per second, far below the 1M+ 4\,KB random
reads per second that an SSD can sustain~\cite{haas2023nvme}.
Reaching this throughput requires hundreds of concurrent outstanding
requests.
Instead of blocking on every read in the naive extension,
we develop an asynchronous pipeline (\S\ref{sec:pipeline}) that overlaps I/O with the computation of scanning the codes in memory.

%% file: content/3-Split-RaBitQ.tex
\section{Split-RaBitQ}
\label{sec:split}

Split-RaBitQ, resolves Problem~1 by
splitting the binary code between memory and SSD:
\S\ref{sec:split:twotier} gives the design, and \S\ref{sec:split:lb}
derives the lower bound and its guarantee from RaBitQ's analysis.

\subsection{The Code Splitting Design}
\label{sec:split:twotier}

Split-RaBitQ keeps RaBitQ's quantization, and what changes is how many
bits of the binary code are kept in memory. %
The first $D_m$ bits of the binary code with $D_m \le D$,
referred to as the \emph{prefix}, stay in memory together with the three factors used to compute the lower bound in the pruning stage.
$D_m$ is fixed at build time according to the RAM budget, and the prefix is packed in batches of 32 for FastScan (\S\ref{sec:pre:rabitq}).
The remaining $D_s := D - D_m$ bits of the binary code, the ex-code, and the two re-ranking factors reside on SSD.

For codeword $\bar{x}$, let $\bar{x}_{D_m}$ be $\bar{x}$ with the remaining $D - D_m$
coordinates set to zero.
Let $\bar{o}_{D_m} := P\bar{x}_{D_m}$. Since
$P$ is orthogonal
and $\bar{x} \in \{\pm 1/\sqrt{D}\}^D$,
we have
$\lVert \bar{o}_{D_m}\rVert^2 = \lVert \bar{x}_{D_m}\rVert^2 = D_m/D$.
Split-RaBitQ estimates $\langle o, q\rangle$ from the prefix by
\begin{equation}
\label{eq:split-est}
\widehat{\beta}_{D_m} =
\frac{D}{D_m} \cdot
\frac{\langle \bar{o}_{D_m}, q\rangle}{\langle \bar{o}, o\rangle} .
\end{equation}
As before, $\langle \bar{o}, o\rangle$ is computed at indexing time
when all coordinates are available, see Eq.~\eqref{eq:alignment}.
The numerator
$\langle \bar{o}_{D_m}, q\rangle = \langle \bar{x}_{D_m}, P^{-1}q\rangle$
is the inner product $\langle \bar{o}, q\rangle$ of
\S\ref{sec:pre:rabitq} restricted to the first $D_m$ coordinates.
The factor $D/D_m$ compensates for the omitted $D - D_m$ coordinates (\S\ref{sec:split:lb}).
As in RaBitQ,
the estimate and lower bound use three precomputed per-vector
factors.
For re-ranking, it fetches the remaining $D - D_m$ bits of the binary code, the ex-code, and the factors from SSD and computes the estimate based on the codes as in the original RaBitQ method (\S\ref{sec:pre:rabitq}).

\begin{figure*}[t]
  \centering
  \includegraphics[width=\textwidth]{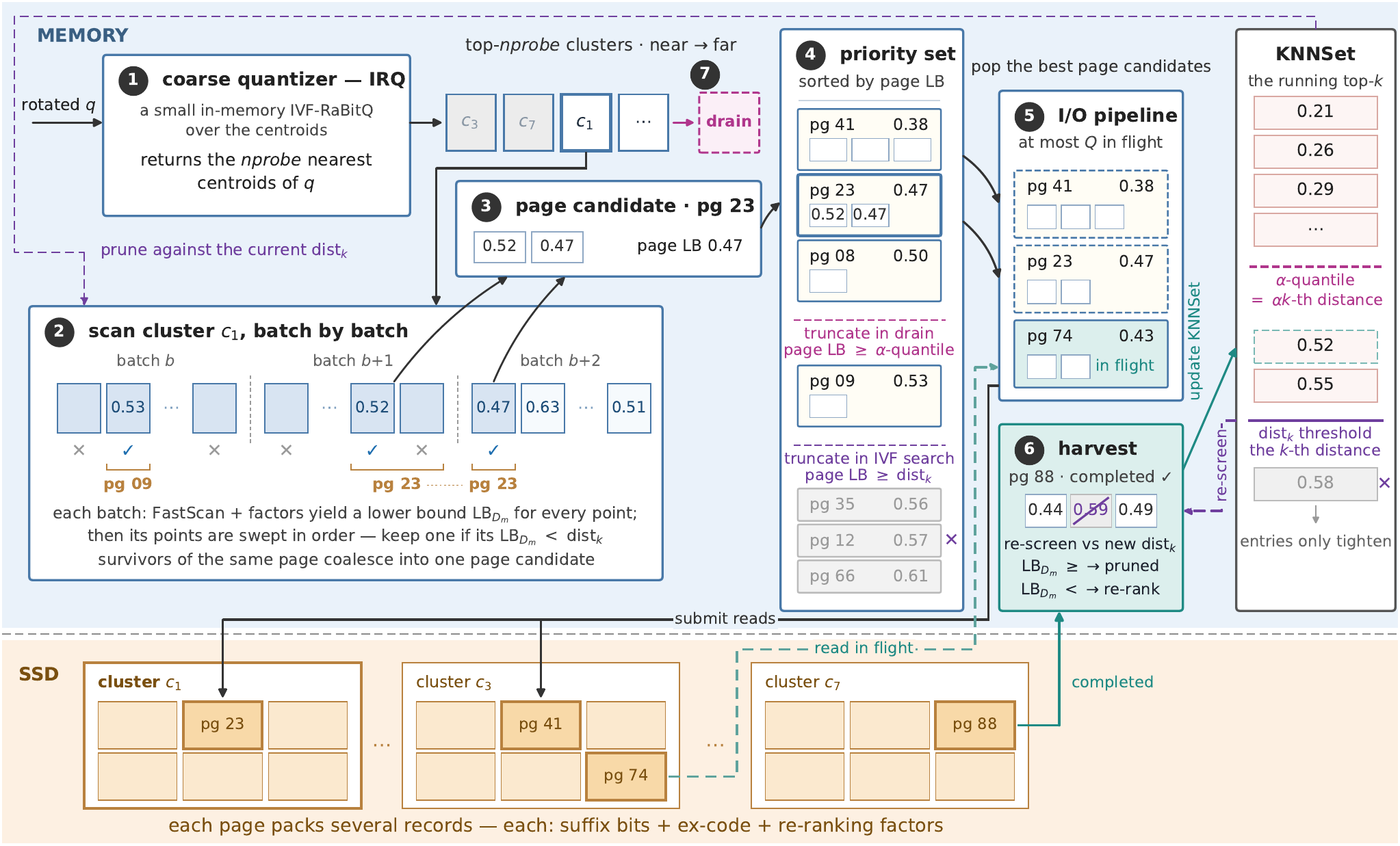}
  \caption{Overview of the \OURS search pipeline. Memory (top) holds the coarse
  quantizer and the per-vector prefix codes; the SSD (bottom) holds one
  record per vector, in cluster-contiguous 4\,KB pages.
  See \S~\ref{sec:pipe:overview} for details on the search process.
  }
  \label{fig:overview}
\end{figure*}

\subsection{A Lower Bound from the Prefix Code}
\label{sec:split:lb}

Let $\beta := \langle o, q\rangle$.
The projection of
$q$ onto $o$ is $\beta o$, and the remaining component $q-\beta o$
is perpendicular to $o$. For non-collinear $o$ and $q$, let $e_1$
be the unit vector along this perpendicular component. Then
$q = \beta o + \sqrt{1-\beta^2}\,e_1$.
For collinear $o$ and $q$, the perpendicular component vanishes.
Substituting this decomposition into Eq.~\eqref{eq:split-est},
the Split-RaBitQ estimation error becomes
\begin{equation}
\label{eq:split-decomp}
\widehat{\beta}_{D_m}
- \beta
= \beta\,\left(\frac{D}{D_m}\frac{\langle \bar{o}_{D_m}, o\rangle}{\langle \bar{o}, o\rangle} - 1\right)
+ \sqrt{1-\beta^2}\;
  \frac{D}{D_m}
  \frac{\langle \bar{o}_{D_m}, e_1\rangle}{\langle \bar{o}, o\rangle} .
\end{equation}

For the full code ($D_m = D$), %
the first term vanishes and the remaining term is RaBitQ's
estimation error.

By Eq.~\eqref{eq:alignment},
$\langle \bar{o}_{D_m}, o\rangle / \langle \bar{o}, o\rangle$ is the
fraction of the $\ell_1$ mass of  $P^{-1}o$ that lies in the prefix
coordinates.
By exchangeability of the coordinates of $P^{-1}o$,  each
coordinate contributes a $1/D$ share of this mass in expectation.
The scaling factor  $D/D_m$ rescales the prefix estimate to the full vector.
The deviation of
the scaled prefix fraction
from one produces the first
error term, while the second comes from the component of
$\bar{o}_{D_m}$ along $e_1$.
Theorem~\ref{thm:split-lb} states the error bound.

\begin{theorem}[Prefix estimator]
\label{thm:split-lb}
$\widehat{\beta}_{D_m}  = \frac{D}{D_m}\frac{\langle \bar{o}_{D_m}, q\rangle}{\langle \bar{o}, o\rangle}$
is an unbiased estimator of $\langle o, q\rangle$. For
$0 < \epsilon \le \sqrt{D_m}/3$, the error bound of the estimator is
given as
\begin{equation*}
\Pr\Biggl\{\,
\Bigl\lvert
\widehat{\beta}_{D_m}
- \langle o, q\rangle
\Bigr\rvert
> \frac{\epsilon}{\sqrt{D_m}}
\,\Biggr\}
\;\le\; 2 e^{-c_1\epsilon^2} ,
\end{equation*}
where $c_1$ is a constant.
\end{theorem}

\noindent
The proof is given in Appendix~\ref{sec:app:proof}.
As in RaBitQ, the error bound can be concisely presented as an
estimation error of $O(1/\sqrt{D_m})$ with high probability.

\noindent\textbf{Lower bound.}
Substituting the estimate $\widehat{\beta}_{D_m} + \epsilon/\sqrt{D_m}$
for $\langle o, q\rangle$ in $d(q_r, o_r)$ gives the lower bound
$\mathrm{LB}_{D_m}(q_r, o_r)$, which lies below the true distance with high probability $1-2e^{-c_1\epsilon^2}$.
Appendix~\ref{sec:app:distribution} examines the
empirical error distribution and the frequency of
lower-bound violations.

Reducing $D_m$ saves memory
and computation during the pruning phase but widens the error bound, potentially increasing
SSD reads and re-ranking work. We evaluate this trade-off in
\S\ref{sec:exp:memdim}.

%% file: content/4-RaBitQ-SSD-pipeline.tex
\section{Search Pipeline}
\label{sec:pipeline}

\subsection{Overview}
\label{sec:pipe:overview}

Figure~\ref{fig:overview} gives an overview of the search pipeline.

\noindent\textbf{Storage.}
The in-memory coarse quantizer routes a query to its $\nprobe$ nearest clusters.
Each cluster stores, for every vector, the $D_m$-bit prefix of
its binary code together with its
three factors (\S\ref{sec:split:twotier}), stored in batches of size 32.
The SSD holds one
fixed-size record per vector, consisting of the suffix bits of the binary code, the
ex-code, and the two re-ranking factors.
The records of a cluster are packed into contiguous 4~kB pages (\S\ref{sec:pre:ssd}), with some padding so that they don't cross page boundaries.
\label{sec:pipe:layout}
The in-RAM and SSD vectors are stored in the same order (by ascending distance to the centroid), so that the page containing a record can be determined directly.

\noindent\textbf{Searching.}
A KNNSet keeps track of the nearest $k$ vectors seen so far and their distances to the query, whose maximum is $\mathrm{dist}_k$.
One query proceeds as follows.
\stp{1}~The rotated query passes the
coarse quantizer, an in-memory IVF-RaBitQ (IRQ) over the cluster
centroids (\S\ref{sec:pipe:coarse}), which returns the $\nprobe$
nearest clusters, visited from near to far.
\stp{2}~Within a cluster,
the scan proceeds by FastScan batches of 32 vectors.
For each vector, the prefix codes yield the lower bound and a
point ``survives'' if it is below the current $\mathrm{dist}_k$.
\stp{3}~Surviving points of the same page
coalesce into a page candidate,
\stp{4}~which enters the priority set
keyed by the page lower bound, the minimum over its members.
\stp{5}~Reads are issued asynchronously from the priority set
while the scan continues.
After each batch, the search polls for completions and issues
additional reads while keeping at most $Q$ reads
in flight per query.
Before issuing a read, the search re-checks the page
candidate against the current $\mathrm{dist}_k$.
\stp{6}~A completed page is harvested: members that
fail the current $\mathrm{dist}_k$ are dropped, and the rest are
re-ranked from the fetched data and may update the KNNSet.
The details of pipelining the inverted-list search will be presented in \S\ref{sec:pipe:invlist}.
\stp{7}~After the
last cluster, the search drains the set and the reads in flight,
popping only candidates whose bound falls below the tighter
$\alpha$-quantile threshold (\S\ref{sec:pipe:alpha}).

\subsection{IRQ: In-Memory IVF-RaBitQ as the Coarse Quantizer}
\label{sec:pipe:coarse}

\begin{figure}[h]
  \centering
  \includegraphics[width=\linewidth]{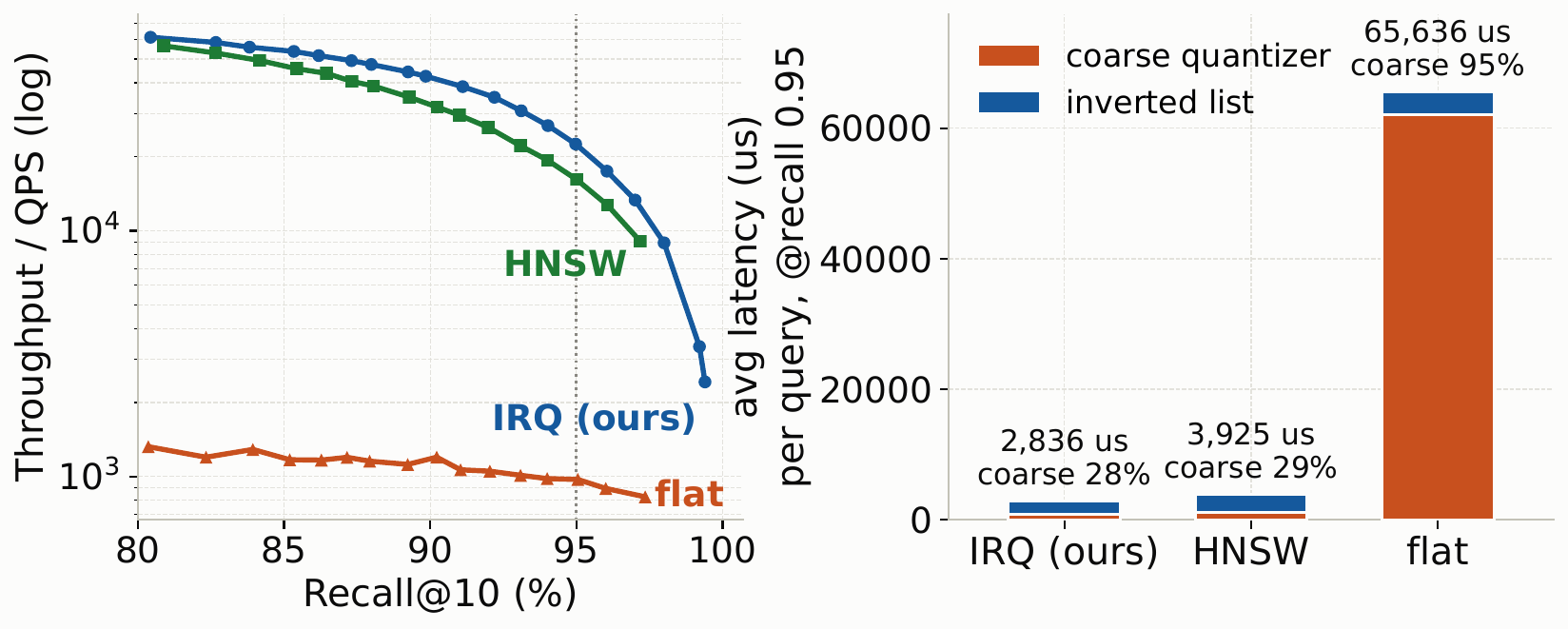}
  \caption{Coarse quantizers on YFCC-100M ($C = 96{,}422$, $k = 10$,
  $T = 64$). Left: throughput against recall. Right: average query
  latency at recall $0.95$, split into the coarse step and the rest
  (``inverted list'').}
  \label{fig:irq}
\end{figure}

Before any inverted list is visited, the coarse quantizer finds the
$\nprobe$ nearest of the $C=\nlist$ centroids.
A flat scan spends $C$
full-precision distance computations on this regardless of $\nprobe$,
while the work that follows scales with $\nprobe$. On
YFCC-100M ($C = 96{,}422$, some $10^3$ points per cluster), the flat scan
accounts for $95\%$ of the query latency at recall@10 $= 0.95$: $62$
of the $66$\,ms spent on a query go to routing alone
(Fig.~\ref{fig:irq}, right).
Building a graph index over the centroids is the standard remedy~\cite{baranchuk2018revisiting,douze2026faisslib}:
with HNSW, average query latency falls to $3.9$\,ms, of which the
coarse step is $29\%$.

We adopt another coarse quantizer, an in-memory IVF-RaBitQ built
over the centroids, \emph{IRQ} for short. Centroids are quantized to
$B = 9$ bits per dimension by default, which is sufficient for high
accuracy~\cite{ExtendedRaBitQ}, and a query retrieves its $\nprobe$
nearest centroids through the prune-and-re-rank flow of
\S\ref{sec:pre:ivfrabitq}. IRQ reduces the average
query latency to $2.8$\,ms ($23\times$ below the flat scan, better than the graph index; Fig.~\ref{fig:irq}, right) and achieves
the best end-to-end performance (Fig.~\ref{fig:irq}, left).
Storing the centroids as $B$-bit codes also keeps the coarse quantizer smaller.

\subsection{Pipelining the Inverted-List Search}
\label{sec:pipe:invlist}

\begin{table*}[t]
  \centering
  \caption{
  Page reads for five I/O pipeline designs of \OURS,
  on YFCC-100M at recall $0.90$
  (per-query means).
  QPS$\times$ is relative to Naive at the same $T$ and $k$.
  \emph{issued}:
  4\,KB page reads submitted. \emph{drop}: page candidates discarded
  in the pending set before any read. \emph{waste}: issued reads whose
  every member the re-screen pruned (zero re-rank).
  \emph{re-rank/surv.}: records re-ranked over points surviving the scan.
  \emph{io-wait}: share of the list-search wall blocked on reads
  (blocking designs). \emph{drain}: share of the same wall spent
  finishing outstanding I/O after the scan (async designs).
  $^{\dagger}$See \S\ref{sec:pipe:alpha}.}
  \label{tab:ladder}
  \small
  \setlength{\tabcolsep}{3pt}
  \begin{tabular}{l rrrrrrr rrrrrrr}
    \toprule
    & \multicolumn{7}{c}{$T=1$ thread} & \multicolumn{7}{c}{$T=64$ threads} \\
    \cmidrule(lr){2-8} \cmidrule(lr){9-15}
    design & QPS$\times$ & issued & drop & waste & re-rank/surv. & io-wait & drain
           & QPS$\times$ & issued & drop & waste & re-rank/surv. & io-wait & drain \\
    \midrule
    \multicolumn{15}{c}{$k = 10$} \\
    \midrule
    Naive      & 1.00 & 102 & --- & --- & 102/102 & 95\% & --- & 1.00 & 102 & --- & --- & 102/102 & 93\% & --- \\
    Coalesced  & 1.30 & 81 & --- & --- & 102/103 & 94\% & --- & 1.30 & 81 & --- & --- & 102/103 & 96\% & --- \\
    Overlapped & 6.67 & 156 & --- & 72 & 104/334 & --- & 4\% & 0.58 & 483 & --- & 400 & 103/833 & --- & 89\% \\
    Scheduled  & 8.35 & 53 & 797 & 8 & 61/2292 & --- & 6\% & 6.98 & 40 & 1204 & 1 & 52/3414 & --- & 21\% \\
    Full       & 8.36 & 54 & 779 & 8 & 61/2247 & --- & 5\% & 7.32 & 38 & 1078\rlap{$^{\dagger}$} & 1 & 50/3056\rlap{$^{\dagger}$} & --- & 10\% \\
    \midrule
    \multicolumn{15}{c}{$k = 100$} \\
    \midrule
    Naive      & 1.00 & 614 & --- & --- & 614/614 & 99\% & --- & 1.00 & 614 & --- & --- & 614/614 & 99\% & --- \\
    Coalesced  & 1.66 & 439 & --- & --- & 614/615 & 98\% & --- & 1.41 & 439 & --- & --- & 614/615 & 98\% & --- \\
    Overlapped & 15.62 & 482 & --- & 39 & 623/750 & --- & 2\% & 1.18 & 1382 & --- & 941 & 615/2773 & --- & 66\% \\
    Scheduled  & 19.94 & 241 & 1432 & 1 & 363/4367 & --- & 13\% & 6.89 & 236 & 6042 & 0 & 357/17040 & --- & 73\% \\
    Full       & 21.89 & 212 & 1495 & 0 & 333/4457 & --- & 5\% & 12.35 & 131 & 5076\rlap{$^{\dagger}$} & 0 & 248/13738\rlap{$^{\dagger}$} & --- & 19\% \\
    \bottomrule
  \end{tabular}
\end{table*}

\OURS scans the $\nprobe$ inverted lists, for which we design several pipelines of increasing complexity, starting from the
naive extension of \S\ref{sec:pre:naive}.
See Table~\ref{tab:ladder} for the experimental setup and I/O numbers cited below.

\newcommand{\pipeline}[1]{\texttt{#1}}

\noindent\textbf{\pipeline{Naive}.}
The naive extension of \S\ref{sec:pre:naive} issues one read per
surviving candidate and blocks on it.
It has one virtue: every pruning
decision sees the tightest $\mathrm{dist}_k$ available so far, because it is fully causal.
The query thread, however,
waits out every read: %
a query issues
$102$ page reads at $k{=}10$ and $614$ at $k{=}100$ on average, each
delivering one useful record out of the three a page holds, and the
thread waits on I/O for $95$--$99\%$ of the list-search time
(Table~\ref{tab:ladder}).

\noindent\textbf{\pipeline{Coalesced}.}
Points of the same page that survive pruning can share one read (\stp{3} in
Fig.~\ref{fig:overview}): because the scan visits the points of a page
consecutively (\S\ref{sec:pipe:layout}), the surviving points of a page
emerge consecutively, and its read is issued when the scan
crosses the page boundary. The speedup is due to fewer page reads:
for $k=10$, $102$ reads per query become $81$, or $1.30\times$ throughput  ($1.66\times$ throughput for $k{=}100$).
The tightest-threshold property of \pipeline{Naive} is preserved at page granularity.

\noindent\textbf{\pipeline{Overlapped}.}
Blocking is what stalls the query thread
(the grey strips of \pipeline{Naive} and \pipeline{Coalesced}
in Fig.~\ref{fig:timeline}), and this step removes it. The scan no
longer waits for reads: the surviving points of a page form a \emph{page
candidate}, whose read is submitted when all points of a page have been
scanned.
After each batch, the search polls for
completed reads without blocking and \emph{harvests} them (\stp{6} in
Fig.~\ref{fig:overview}): every member is re-screened against the
current $\mathrm{dist}_k$, and those that pass are re-ranked.
When the last cluster has been scanned, the
search keeps harvesting until every outstanding read
returns, i.e., the \emph{drain} (\stp{7} in Fig.~\ref{fig:overview}).
As $\mathrm{dist}_k$ may have tightened
since the read was submitted, some pages may have been fetched for nothing,
which Table~\ref{tab:ladder} reports as \emph{waste}.
How \pipeline{Overlapped}
performs depends on whether the SSD is saturated. Below its IOPS
ceiling ($T{=}1$), the overlap is a clear gain: over \pipeline{Coalesced},
throughput rises another $5.1\times$ at $k{=}10$ ($9.4\times$ at
$k{=}100$), even though $156$ reads are issued where \pipeline{Coalesced} needed
$81$.
These $46\%$ wasted reads are acceptable
because the benefit of overlapping I/O outweighs
the cost of the extra reads.
Once the device is saturated
($T{=}64$, where the $64$ threads together drive the SSD at its IOPS
ceiling), additional reads begin to queue.
I/O completions slow down, delaying both harvesting and tightening of $\mathrm{dist}_k$.
The resulting stale threshold allows more candidates to survive pruning and triggers still more reads.
In our trace, the number of outstanding reads for one query peaks at $677$ in Fig.~\ref{fig:timeline}b.
In this regime a query %
wastes $83\%$ of the reads  and
spends $89\%$ of its list-search time in the drain (Table~\ref{tab:ladder}).
Throughput falls below
\pipeline{Coalesced} at both $k$ ($0.45\times$ and $0.84\times$) and, at
$k{=}10$, below even the blocking \pipeline{Naive}.

\noindent\textbf{\pipeline{Scheduled}.}
This step bounds the number of reads in flight and moves each read
decision from scan time to issue time, when the threshold is tighter.
Instead of having its read submitted as soon as the scan leaves the
page, a page candidate is deposited in the \emph{priority set}
(\stp{4} in Fig.~\ref{fig:overview}), ordered by page lower bound.
Reads are issued in that order, keeping at most $Q$ in flight
(\stp{5}). After each batch, the search polls and harvests as before, and conditionally submits new reads by order of lower bound:
(a) if fewer than $Q$ are in flight and
(b) if each candidate lower bound is lower than the current
$\mathrm{dist}_k$.
The issuance stops at the first candidate that fails check (b) as the remaining candidates would fail too
or the set is emptied.
While a
candidate waits in the set the threshold keeps tightening, and most
candidates are pruned there without a read (\emph{drop} column of
Table~\ref{tab:ladder}).
Because Scheduled keeps fewer reads in flight than Overlapped, its
threshold is slightly more stalled and more points survive
the scan, but many page candidates are still dropped before incurring I/Os.
Fewer fetched pages also mean less harvest
work, visible in Fig.~\ref{fig:timeline}b as the shorter blue bar of
Scheduled: \pipeline{Overlapped} fetches $483$ pages to \pipeline{Scheduled}'s $40$ ($103$
vs.\ $52$ re-ranks).

\begin{figure*}[h]
  \centering
  \includegraphics[width=0.65\linewidth]{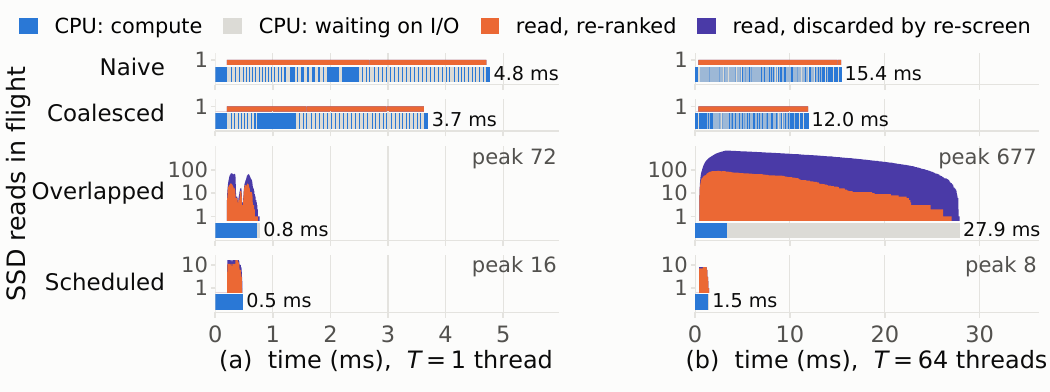}
  \caption{Measured timeline of one query under each pipeline, at
  (a) $T{=}1$ and (b) $T{=}64$ (YFCC-100M, $k{=}10$, recall $0.90$; the
  median-latency query of a full 10k-query run).
  Per pipelining step, the lower
  strip shows the CPU state over time, compute (blue) or waiting on
  I/O (grey), and the area above counts SSD reads in flight (log
  scale), split into reads whose page fed a re-rank (orange) and reads
  discarded whole by the re-screen (violet). Annotations give the peak
  in-flight count and the total latency.}
  \label{fig:timeline}
\end{figure*}

\noindent\textbf{Discussion.}
Figure~\ref{fig:timeline} shows the SSD reads in flight and CPU activity for each of the pipelines, depending on the number of threads $T$.
Below the IOPS ceiling ($T{=}1$), the gain over \pipeline{Overlapped} is modest,
$1.25\times$ at $k{=}10$ and $1.28\times$ at $k{=}100$, because
wasted reads add little wall time when the device is not saturated. The
number of reads, however, falls from $156$ to $53$ per query, fewer
than even the blocking \pipeline{Coalesced} ($81$).

Once the device is saturated ($T{=}64$), the
advantage is decisive: $40$ reads per query, essentially none wasted,
and the in-flight count stays flat at $Q$ ($8$ here, $16$ at $T{=}1$)
instead of rising into the hundreds (Fig.~\ref{fig:timeline});
the traces show $27.9$\,ms for \pipeline{Overlapped}
and $1.5$\,ms for \pipeline{Scheduled}, and
throughput is $12\times$ \pipeline{Overlapped} at $k{=}10$ and $5.8\times$ at
$k{=}100$ (Table~\ref{tab:ladder}).

We further evaluate the effect of lower-bound ordering
by replacing it with FIFO while keeping all other settings
of \pipeline{Scheduled} unchanged.
At $T=64$ and recall $0.90$, FIFO issues roughly twice as
many reads per query, reducing throughput to $43\%$ and
$53\%$ of \pipeline{Scheduled}'s for $k=10$ and $k=100$, respectively.

\subsection{$\alpha$-Quantile Threshold to Alleviate Drain}
\label{sec:pipe:alpha}

\noindent\textbf{A drainage problem.}
The drain runs after the last cluster, as no scanning is left to run in parallel with reads. %
Below the IOPS
ceiling ($T{=}1$) reads complete promptly, little remains after the
scan, and the drain represents $6$--$13\%$ of the list-search time
(Table~\ref{tab:ladder}).
Once the device is saturated ($T{=}64$),
completions slow down (in our traces the median read completes in
$186$\,\textmu{}s), %
so
fewer reads are issued during the scan under the bound $Q$,
more page candidates are left  when it ends,
and the drain takes $21\%$ of the list-search time at
$k{=}10$ and even $73\%$ at $k{=}100$.

\noindent\textbf{\pipeline{Full.}}
This pipeline issues
fewer reads by tightening the drain's threshold.
Once the
last cluster is scanned (\stp{7} in Fig.~\ref{fig:overview}), a
pending candidate is issued only if it could still place among the
best $\alpha k$ of the current results, with $\alpha \in (0,1]$ and $\alpha k$ an integer.
This means its lower bound falls below the \emph{$\alpha$-quantile of the
KNNSet}.
$\alpha=1$ recovers $\mathrm{dist}_k$ and hence \pipeline{Scheduled}, and $\alpha < 1$
gives a tighter threshold. The tighter threshold applies only to reading decisions; harvested
pages are still re-screened against $\mathrm{dist}_k$.

The threshold comes with a guarantee:
with exact re-ranking and valid lower bounds,
every point among the true top-$\alpha k$ keeps
its lower bound below the $\alpha$-quantile threshold and is never pruned.
In this way, the search guarantees a recall of at least $\alpha$ over the points
the probed clusters contain.
In practice a far larger fraction of the true top-$k$ survives the $\alpha$-quantile.
We measured this in the tightest possible setting on
YFCC-100M, with the running KNNSet replaced by the true top-$k$ so
that the threshold is the $\alpha k$-th smallest true
distance: at $\alpha = 0.4$, about $85\%$ of the true top-$k$
neighbors have a lower bound below this threshold and would not be
pruned ($84.2\%$ at $k{=}10$, $85.1\%$ at $k{=}100$), against the
$40\%$ the guarantee promises, and we adopt $0.4$ as the default.

\noindent\textbf{Impact on recall and determinism.}
In the actual search,
the recall loss incurred by the $\alpha$-quantile threshold
is even smaller than suggested by this measurement, for three reasons: (1) the running KNNSet typically gives a looser threshold than the true top-$k$; (2) only the candidates still pending when the drain
begins face the $\alpha$-quantile threshold; and (3) pruning is per page candidate,
against the smallest lower bound among its members, so one strong
surviving point keeps its whole page.

We note that, with exact re-ranking and valid lower
bounds, I/O completion order does not affect the top-$k$ set
returned by \pipeline{Scheduled} ($\alpha=1$),
but can affect that returned by \pipeline{Full} ($\alpha<1$).
In \pipeline{Full}, depending on the order reads finish, re-ranking these candidates can tighten the
$\alpha$-quantile threshold sooner or later.
Thus, I/O completion order can affect which neighbors
are retained at ranks $(\alpha k,k]$.
Appendix~\ref{sec:app:determinism} shows
that the effect of I/O completion order on the returned
neighbor sets and recall is low.

The gain follows the drain. Below the IOPS ceiling ($T{=}1$), where
the drain is light, \pipeline{Full} matches \pipeline{Scheduled} at $k{=}10$ ($1.00\times$)
and gains $1.10\times$ at $k{=}100$. Once the device is saturated
($T{=}64$), \pipeline{Full} reads $131$ pages per query at $k{=}100$ where
\pipeline{Scheduled} reads $236$, the drain share falls from $73\%$ to $19\%$,
and throughput rises $1.8\times$.
As every query's drain issues fewer reads,
contention on the SSD decreases, allowing
scan-phase reads to complete sooner and the scan threshold
to tighten earlier,
so fewer points survive the scan (the daggered counts of surviving points
in Table~\ref{tab:ladder}: $13{,}738$ versus $17{,}040$ at
$k{=}100$), which leaves fewer pending candidates to drain.

%% file: content/5-experiment.tex
\begin{table}[!htbp]
  \centering
  \caption{Datasets. All vectors are \texttt{float32}
  except those in DINO-10B, which are \texttt{uint8};
  the columns give the modality, the dimensionality $D$,
  base count, the search metric and whether queries are distributed with the dataset.}
  \label{tab:datasets}
  \small
  \setlength{\tabcolsep}{4.5pt}
  \begin{tabular}{llrrcr}
    \toprule
    Dataset & Modality & $D$ & \#Base & Metric & queries? \\
    \midrule
    OpenAI-5M~\cite{openai2024embedding} & Text & 1536 & 5\,M & $\ell_2$ & no \\
    Wiki-10M~\cite{cuvs_wikiall} & Text & 768 & 10\,M & $\ell_2$ & yes \\
    DPR-100M~\cite{intel_dpr_dataset} & Text & 768 & 101.2\,M & IP & no \\
    LAION-100M~\cite{sisap2024laion100m} & Image & 768 & 102.1\,M & $\ell_2$ & yes \\
    YFCC-100M~\cite{bigann_yfcc_images} & Image & 1280 & 98.7\,M & $\ell_2$ & yes \\
    DataComp-1B~\cite{datacomp1B} & Image & 512 & 1\,B & $\ell_2$ & no\\
    DINO-10B~\cite{dino10b}
      & Image
      & 1024
      & 10\,B
      & $\ell_2$ & yes\\
    \bottomrule
  \end{tabular}
\end{table}

\begin{figure*}[!htbp]
  \centering
  \includegraphics[width=0.9\textwidth]{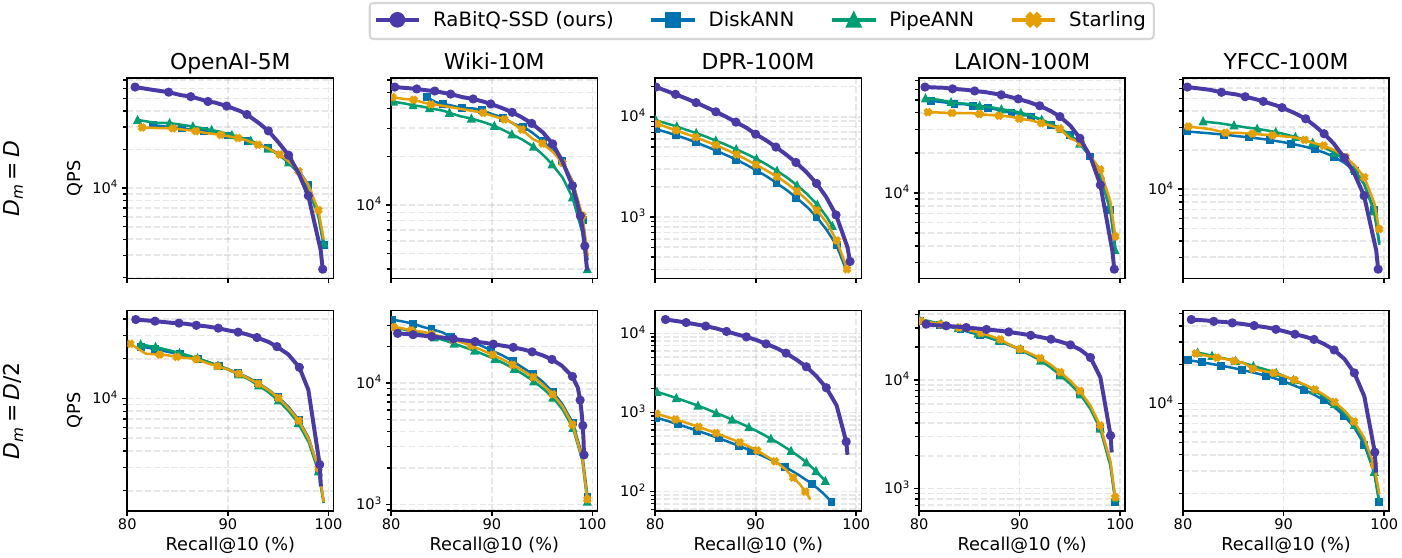}
  \caption{Main evaluation at $k = 10$ on the five datasets (columns):
  throughput against recall at $T = 64$ (log scale, up is better),
  under two memory budgets, $D_m = D$ (top row) and $D_m = D/2$
  (bottom row), with baselines using matched per-point
  memory budgets (\S\ref{sec:exp:setup}).}
  \label{fig:main}
\end{figure*}

\section{Experimental Study}
\label{sec:experiment}

\subsection{Experimental Setup}
\label{sec:exp:setup}

\newcommand{\rmax}{$r_\mathrm{max}$}

\noindent\textbf{Platform.}
All experiments run on a server with an AMD Ryzen Threadripper PRO
9975WX ($32$ cores, $64$ threads), $768$\,GB of DDR5 memory
($8\times96$\,GB, $8$ channels), and a Samsung PM9D3a $15.36$\,TB
NVMe SSD attached over PCIe~4.0~$\times$4, under Ubuntu 24.04
(Linux 6.17); all binaries are compiled with GCC 13.3 at
\texttt{-O3 -march=native}.
We measured that this configuration
sustains \rmax=$1.63$M random $4$\,KB reads per second.
We use this rate as the device reference in the I/O analysis.
All SSD accesses, ours and every
baseline's, go through \texttt{libaio} under
\texttt{O\_DIRECT}~\cite{open2manpage}, bypassing the page cache.
Each query vector is run on one thread, and the number  $T\le 64$ of concurrent  threads depends on the experiment.

\smallskip\noindent\textbf{Datasets and queries.}
Table~\ref{tab:datasets} lists the seven datasets used
in our evaluation.\footnote{All data access and
experiments were performed by Nanyang Technological University.}
The first six datasets are evaluated with
$10{,}000$ queries.
For datasets that do not provide a dedicated query set
we sample queries disjoint from the base.
Ground truth is computed by exhaustive search.
DINO-10B and its evaluation setup are described
separately in \S\ref{sec:exp:scale}.

\smallskip\noindent\textbf{Our configuration.}
Unless stated otherwise, \OURS uses
$B = 9$ bits per dimension (i.e., $8$-bit ex-codes)
and $C = \nlist = N/1024$ clusters.
Coarse quantization uses IRQ (\S\ref{sec:pipe:coarse}),
which encodes the $C$ centroids at $9$ bits per
dimension and organizes them into $C/1000$ inverted lists.
It probes $\gamma \cdot \nprobe$ of these lists
with $\gamma = 0.5$.
The drain threshold is $\alpha = 0.4$ (\S\ref{sec:pipe:alpha});
the per-thread in-flight I/O budget is $Q = 16$ pages for $T \le 16$ and $Q = 8$
for $T > 16$ (\S\ref{sec:pipe:invlist}).
We conduct the parameter study on the in-flight budget $Q$ and the drain threshold $\alpha$ in Appendix~\ref{sec:exp:param}.

\smallskip\noindent\textbf{Baseline methods.}
In the main evaluation, we compare against three graph-based SSD-resident systems:
DiskANN~\cite{subramanya2019diskann,diskannrepo},
Starling~\cite{wang2024starling,starlingrepo}, and
PipeANN~\cite{guo2025pipeann,pipeannrepo}.
We follow prior work~\cite{chen2026alayalaser,guo2025pipeann}
in configuring these graph baselines and use their default
on-disk representations, which store full-precision vectors.
Each disk graph is built with maximum out-degree $R = 64$ and
construction list size
$L = 200$; Starling and PipeANN build their in-memory navigation
graphs from $1\%$ data samples; DiskANN's node cache is provisioned
with $1\%$ of the raw dataset size; and the search beam widths are $8$,
$8$, and $32$ for DiskANN, Starling, and PipeANN,
respectively.
We report comparisons with
SPANN~\cite{chen2021spann} and
AlayaLaser~\cite{chen2026alayalaser,alayaliterepo}
in \S\ref{sec:exp:spann}, together with their configurations.

\smallskip\noindent\textbf{Memory budgets.}
We consider two memory regimes: low memory sets $D_m=D/2$ and high memory is $D_m=D$.
We equalize the per-vector in-memory storage across
\OURS and the three graph baselines. For each setting of $D_m$, the
per-point PQ bytes of DiskANN, Starling, and PipeANN equal the
per-point size of our in-memory prefix code and factors,
i.e., $D_m/8+12$ bytes.
All indexes are built with the same temporary memory budget: $128$\,GB per dataset, and $512$\,GB for DataComp-1B.

\smallskip\noindent\textbf{Metrics.}
We report \textbf{recall@$k$}, per-query latency at $T = 1$
and throughput (queries per second, \textbf{QPS}) for $T>1$.
When QPS is indicated in plots, it is reported in log scale.
To quantify SSD traffic, we also report the mean number
of $4$\,KB \textbf{pages read per query}.
We aggregate this number into the \textbf{read rate} (pages / s), obtained by multiplying the QPS with the average number of page reads per query.
The read rate is capped by \rmax.
Curves sweep each method's accuracy parameter ($\nprobe$ for \OURS,
the search list size for the graph baselines).
For comparisons at fixed recall, we pick the smallest setting of $\nprobe$ or search list size that achieves this recall.

\subsection{Main Evaluation}
\label{sec:exp:main}

Figure~\ref{fig:main} shows that \OURS achieves
the highest throughput on all five datasets at recalls $0.90$
and $0.95$ under both memory budgets.
\OURS's advantage is strongest for the limited memory budget ($D_m=D/2$) at recall $0.95$, where the speedup is $1.64$--$2.52\times$ w.r.t. the best baselines, and even $16.18\times$ for DPR-100M.

\noindent {\bf Page reads.}
Table~\ref{tab:io} reports the pages read per query.
\OURS reads fewer pages per query than every baseline.
The lowest baseline page count is still $1.1$--$3.8\times$ ours.
We combine in-memory pruning and page-level
I/O scheduling with compact on-disk records, which reduce the number of page reads issued per query and increase the observed throughput.
On the four
datasets other than DPR-100M, and for all methods, the read rate is
$1.4$--$1.7$M/s,
close to \rmax. %
When the device is saturated,
fewer pages per query translate
into higher throughput:
on YFCC-100M, our $37.8$ pages against
PipeANN's $62.2$ accompany a $1.58\times$ throughput advantage.

\begin{table}[!htbp]
\caption{4\,KB pages read per query at recall $0.90$ ($k = 10$,
$T = 64$, $D_m = D$). In bold: the lowest (best) per column.}
\label{tab:io}
\centering
\small
\setlength{\tabcolsep}{3.5pt}
\begin{tabular}{@{}lrrrrr@{}}
\toprule
 & OpenAI & Wiki & DPR & LAION & YFCC \\
\midrule
\OURS & \textbf{37.5} & \textbf{37.9} & \textbf{112.6} & \textbf{32.3} & \textbf{37.8} \\
DiskANN    & 64.1 & 41.2 & 555.7 & 40.4 & 72.3 \\
PipeANN    & 62.3 & 51.5 & 426.5 & 38.9 & 62.2 \\
Starling   & 63.3 & 43.0 & 489.3 & 38.0 & 63.1 \\
\bottomrule
\end{tabular}
\end{table}

\begin{figure*}[!b]
  \centering
  \includegraphics[width=0.9\textwidth]{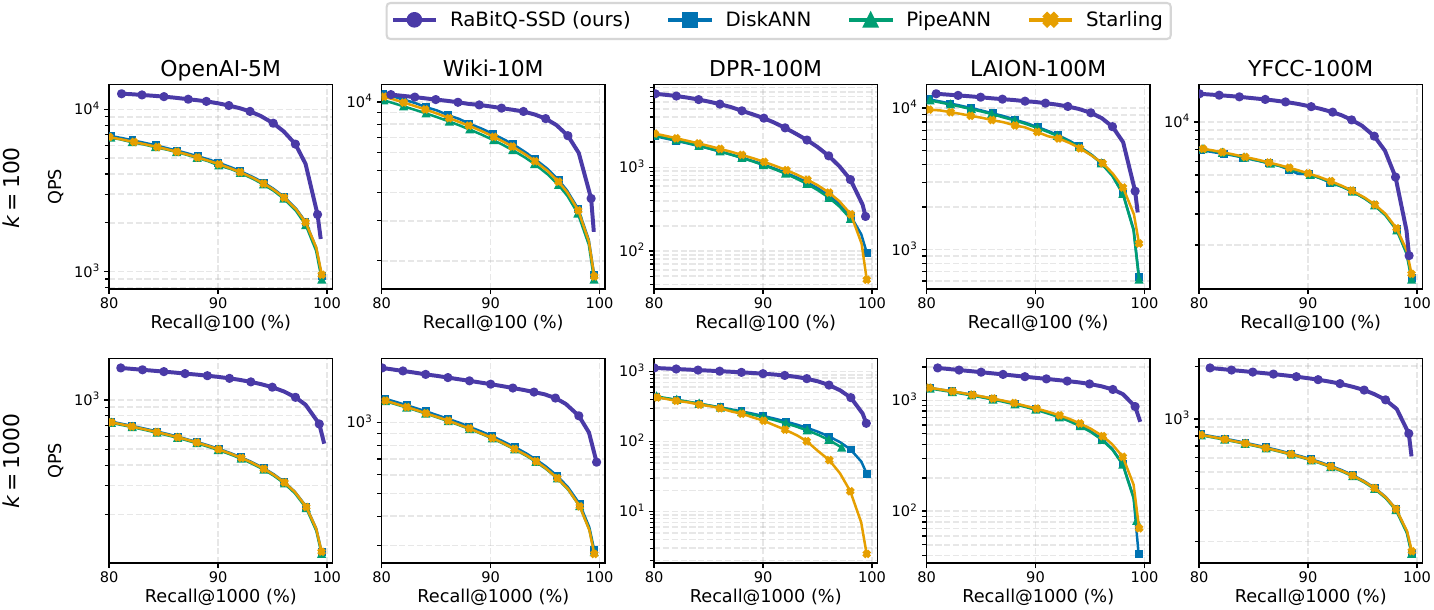}
  \caption{Throughput against recall on the five datasets
  (columns) at $T=64$ and $D_m=D$, for $k=100$ (top row)
  and $k=1000$ (bottom row). The QPS axes use a log scale.}
  \label{fig:topk}
  \Description{QPS versus recall for \OURS, DiskANN, PipeANN,
  and Starling. Columns show OpenAI-5M, Wiki-10M, DPR-100M,
  LAION-100M, and YFCC-100M. The top row shows k=100 and the
  bottom row shows k=1000. QPS axes are logarithmic.}
\end{figure*}

\begin{table*}[!t]
\caption{Per-query latency (ms) at $T = 1$: P50\,/\,P95\,/\,P99 at
recall $0.90$, $k = 10$, under both memory budgets.
}
\label{tab:latency}
\centering
\small
\setlength{\tabcolsep}{3.2pt}
\begin{tabular}{@{}cl ccccc@{}}
\toprule
 & & OpenAI-5M & Wiki-10M & DPR-100M & LAION-100M & YFCC-100M \\
\midrule
\multirow{4}{*}{$D_m = D$} & \OURS & 0.61\,/\,\textbf{0.79}\,/\,\textbf{0.86} & \textbf{0.64}\,/\,\textbf{0.80}\,/\,\textbf{0.87} & 3.79\,/\,4.27\,/\,4.52 & \textbf{0.53}\,/\,\textbf{0.69}\,/\,\textbf{0.76} & \textbf{0.63}\,/\,\textbf{0.79}\,/\,\textbf{0.87} \\
 & DiskANN    & 0.80\,/\,1.05\,/\,1.22 & 0.81\,/\,1.09\,/\,1.26 & 8.72\,/\,9.47\,/\,10.5 & 0.76\,/\,1.05\,/\,1.25 & 0.87\,/\,1.22\,/\,1.42 \\
 & PipeANN    & \textbf{0.60}\,/\,0.82\,/\,0.96 & \textbf{0.64}\,/\,0.95\,/\,1.15 & \textbf{2.63}\,/\,\textbf{3.02}\,/\,\textbf{3.25} & 0.62\,/\,0.89\,/\,1.12 & 0.67\,/\,0.94\,/\,1.11 \\
 & Starling   & 0.91\,/\,1.22\,/\,1.42 & 0.83\,/\,1.10\,/\,1.29 & 7.71\,/\,8.45\,/\,9.14 & 0.82\,/\,1.13\,/\,1.42 & 0.97\,/\,1.34\,/\,1.56 \\
\midrule
\multirow{4}{*}{$D_m = D/2$} & \OURS & \textbf{0.66}\,/\,\textbf{0.89}\,/\,\textbf{1.04} & \textbf{0.76}\,/\,\textbf{1.01}\,/\,\textbf{1.27} & \textbf{3.80}\,/\,\textbf{5.23}\,/\,\textbf{5.93} & \textbf{0.62}\,/\,\textbf{1.00}\,/\,1.53 & \textbf{0.68}\,/\,\textbf{0.90}\,/\,\textbf{1.00} \\
 & DiskANN    & 0.98\,/\,1.23\,/\,1.40 & 1.40\,/\,1.67\,/\,1.87 & 78.6\,/\,83.4\,/\,86.4 & 1.34\,/\,1.64\,/\,1.89 & 1.09\,/\,1.43\,/\,1.64 \\
 & PipeANN    & 0.70\,/\,0.92\,/\,1.07 & 0.85\,/\,1.12\,/\,1.32 & 17.0\,/\,18.9\,/\,20.2 & 0.82\,/\,1.11\,/\,\textbf{1.37} & 0.78\,/\,1.06\,/\,1.25 \\
 & Starling   & 1.12\,/\,1.42\,/\,1.62 & 1.51\,/\,1.77\,/\,1.97 & 97.9\,/\,103.2\,/\,105.8 & 1.39\,/\,1.74\,/\,2.05 & 1.23\,/\,1.61\,/\,1.86 \\
\bottomrule
\end{tabular}
\end{table*}

\noindent {\bf High recall operating points.}
At recall $0.98$ and when
$D_m=D$, the fastest graph baseline overtakes
\OURS on each of the four datasets other than DPR-100M .
In this regime, fewer page reads do \emph{not} always
yield higher throughput. On Wiki-10M, \OURS \emph{still} reads fewer pages per query than
DiskANN ($73.2$ versus $108.7$), and sustains a lower page-read
rate ($0.97$M versus $1.54$M per second).
However, reaching this recall requires a larger $\nprobe$,
increasing in-memory scanning work, so computation becomes the limiting factor for throughput. %

\noindent {\bf Smaller memory budget.}
At $D_m=D/2$ and recall $0.98$, however, \OURS achieves
$2.11$--$2.80\times$ the throughput of the graph baselines
on the four datasets other than DPR-100M.
Reducing $D_m$ loosens the lower bounds, allowing
more candidates to survive pruning.
On Wiki-10M, \OURS reads $141.7$ pages per query and
sustains $1.61$M page reads per second, close to the device limit.
The increased I/O limits throughput to about $11.4$K QPS,
below the $13.3$K achieved at $D_m=D$.
Similarly, for the graph baselines, shorter in-memory PQ codes
yield less accurate distance estimates,
requiring more graph exploration and page reads to reach
the same recall:
DiskANN reads $329.3$ pages per query and achieves only
$4.8$K QPS, leaving \OURS with a $2.37\times$ advantage.

\noindent {\bf The DPR-100M outlier.}
The larger throughput advantage on DPR-100M at
the smaller memory budget reflects a particularly large
increase in the graph baselines' I/O cost.
At recall $0.90$, shortening their in-memory PQ codes
from $108$ to $60$ bytes increases page reads per query
by $6.5$--$9.6\times$.
Their page-read rates remain close to $1.6$M per second,
so throughput falls by $6.5$--$9.9\times$.
In comparison, \OURS resists better: page reads per query increase
by a factor of only $1.73$ (from $112.6$ to $194.8$).
The shorter prefix also reduces FastScan work, allowing
throughput to increase by $24\%$ despite the additional I/O
(\S\ref{sec:exp:memdim}).
Consequently, \OURS's throughput advantage over the fastest
baseline increases from $1.74\times$ at $D_m=D$ to
$14\times$ at $D_m=D/2$.

\noindent {\bf Latency variance.} Tail latency matters at production
request volumes, as even a small fraction of slow requests translates
to a large absolute number.
Table~\ref{tab:latency} reports P50, P95, and P99 per-query
latencies (i.e., the 50\%, 95\% and 99\% latency quantiles) at recall $0.90$ with $T=1$ under both memory budgets.
\OURS also achieves the lowest P95 and P99 latencies in most settings.

\subsection{Performance under Varying Operating Conditions}
\label{sec:exp:conditions}

\smallskip\noindent\textbf{Increasing $k$.}
Figure~\ref{fig:topk} extends the throughput comparison
to $k=100$ and $k=1000$ on all five datasets.
\OURS maintains higher throughput than the baselines.
At recall $0.90$, its advantage over the fastest baseline
increases with $k$ on every dataset, reaching
$1.31$--$3.31\times$ at $k=100$ and
$1.65$--$4.01\times$ at $k=1000$.
Raising $k$ increases the computations and SSD page reads required to
retrieve more neighbors: for \OURS, the larger distance
threshold allows more candidates to survive pruning, and graph
baselines need to explore more nodes to find a larger result set.
For example, on YFCC-100M at recall $0.90$,
\OURS reads $37.8$, $132$, and $957$ pages per query at
$k = 10$, $100$, and $1000$ while the
fastest baseline at each $k$ reads $62.2$, $326$, and $2{,}803$.
Since in this regime, the SSD throughput saturates, this translates to a larger throughput advantage for \OURS: from $1.58\times$ at $k = 10$ to
$2.42\times$ at $k = 100$ and then to $2.86\times$ at $k = 1000$.
\OURS benefits from the IVF data layout:
the candidates to be fetched
are more likely to be co-located in the same SSD page than with graph based methods.

\smallskip\noindent\textbf{Query concurrency.}
An ANN service must use its available SSD read
capacity efficiently, particularly when it shares the device with
other service instances on the same machine. We examine how
throughput changes as increasing query concurrency drives the
SSD from low utilization toward saturation.

Figure~\ref{fig:iobudget}(a) reports the throughput against the read rate.
At low query concurrency, the achieved read rates
are well below \rmax. More concurrent
queries can sustain more outstanding I/O requests, better
exploiting the SSD's internal parallelism and raising both its
utilization and query throughput.
Once
the device is saturated,
additional query concurrency yields little further increase in
the SSD read rate.
Query throughput is then limited by \rmax divided by the number of pages read per query.
Additional query concurrency can increase I/O
queuing delays and hence per-query latency.
At $T=64$, all four methods sustain $1.63$--$1.67$M page reads
per second, close to \rmax.
\OURS reads about $132$ pages per query against the
baselines' $326$--$329$, accounting for its $12{,}371$ QPS
against their $5{,}049$--$5{,}118$.
The graph baselines approach saturation at lower
query concurrency, so most points cluster near the right
end of the curves.
For these baselines, increasing $T$ from $32$ to $64$
hardly increases throughput, while mean query latency
approximately doubles, from $6.3$--$6.4$\,ms to
$12.4$--$12.6$\,ms.

\begin{figure}[t]
  \centering
  \includegraphics[width=0.9\linewidth]{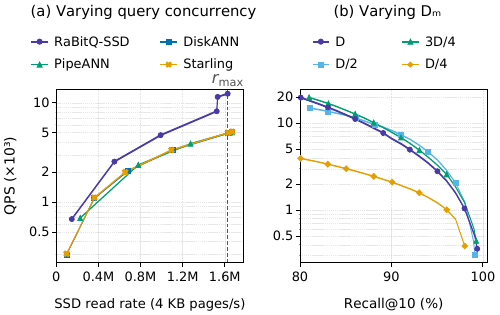}
  \caption{
  Throughput (QPS, log scale) under varying operating conditions.
\newline
{(a) Varying query concurrency} on YFCC-100M
($D_m=D$, $k=100$, recall $0.90$).
Curves are obtained by sweeping $T\in\{1,4,8,16,32,64\}$.
Vertical line = maximal SSD read rate \rmax.
\newline
(b) Varying $D_m$ %
on DPR-100M
($D=768$, $k=10$, $T=64$)
}
  \label{fig:io-memdim}
  \label{fig:iobudget}
  \label{fig:memdim}
\end{figure}

\smallskip\noindent\textbf{Varying the in-memory size.}
\label{sec:exp:memdim}
Figure~\ref{fig:memdim}(b) compares Split-RaBitQ with
$D_m \in \{D/4,\, D/2,\, 3D/4,\, D\}$. %
The per-point in-memory footprint is $D_m/8 + 12$ bytes;
reducing $D_m$ from $D$ to $D/4$ shrinks it
from $108$ to $36$ bytes ($3.0\times$), and moves more data to SSD.
This cuts the FastScan work per candidate but loosens the bound
(width${}\propto 1/\sqrt{D_m}$, \S\ref{sec:split:lb}),
increasing SSD reads and re-ranking work for the same recall.

On DPR-100M at recall $0.90$, search at $D_m=D$
is compute-bound, with an SSD read rate of only $0.75$M pages
per second.
Reducing $D_m$ relieves this computational
bottleneck. Although the looser bound increases reads per query,
the SSD can accommodate the additional I/O,
allowing throughput to improve. At $3D/4$, the read rate is still
only $1.13$M pages per second; at $D/2$, it reaches $1.62$M,
close to \rmax.
Throughput at $D/2$ reaches $8{,}298$ QPS, $24\%$ above $D$
and the highest among the four settings, despite the increase
from $113$ pages per query at $D_m=D$ to $195$ at $D_m=D/2$.
With the SSD near saturation at $D_m=D/2$, further shortening FastScan
provides little benefit: reducing $D_m$ to $D/4$ increases the
pages to read to $769$, so the throughput consequently falls to
$2{,}117$ QPS.

 \subsection{Scaling to Billions of Points}
\label{sec:exp:scale}

We evaluate \OURS's behavior closer to the memory and disk capacity of our test machine.

\begin{figure}[h]
  \centering
  \includegraphics[width=\linewidth]{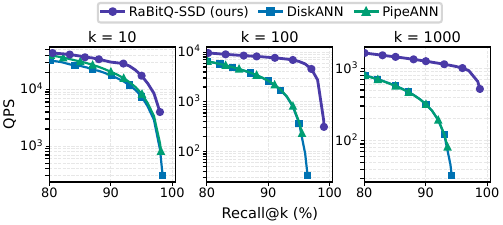}
  \caption{DataComp-1B: throughput against recall at $T = 64$ and
  $D_m = D = 512$, for varying $k$.}
  \label{fig:billion}
\end{figure}

\smallskip\noindent\textbf{1B scale.}
\label{sec:exp:billion}
We evaluate this scale on DataComp-1B.
We set $D_m = D = 512$, so
the in-memory codes and factors occupy
$76$ bytes per point.
We compare against DiskANN and PipeANN with per-point PQ size
matched at the same $76$ bytes.
Figure~\ref{fig:billion} shows the comparison at $T = 64$. %
\OURS retains its throughput advantage at this scale.
At recall $0.90$, its throughput is $1.48\times$, $2.96\times$,
and $3.9\times$ that of the fastest baseline for $k=10$, $100$,
and $1000$, respectively.
As in \S\ref{sec:exp:conditions}, the advantage increases with $k$.

\smallskip\noindent\textbf{10B scale.}
\label{sec:exp:tenb}
DINO-10B~\cite{dino10b} holds $10^{10}$ vectors of $1024$ dimensions
in \texttt{uint8} under $\ell_2$, and comes with $100{,}000$ queries
and ground truth at $k = 10$, on which we evaluate at $T = 64$.
With $8$-byte point IDs, the full-prefix codes, factors,
and IDs alone require $1.48$\,TB, exceeding our server's $768$\,GB
of DRAM. Setting $D_m=256$ reduces the resident index to
$545$\,GB including the coarse quantizer (with $C \approx N/4096$ clusters).
We store the raw uint8 vectors on SSD instead of
ex-codes for re-ranking, occupying $10.24$\,TB.
We do not evaluate DiskANN or PipeANN at this scale, as
10B nodes cannot be represented in their 32-bit internal node
IDs~\cite{diskannrepo,pipeannrepo}.
An extension to 64-bit IDs, retaining
their existing disk layouts,
would require an estimated $20.48$\,TB for the main disk
index alone, exceeding our SSD's capacity.
The bulkiness of the edge storage is a fundamental limitation of graph based methods.

\begin{figure}[h]
  \centering
  \begin{tabular}{@{}c@{\hspace{3.8pt}}c@{}}
  (a) QPS vs. recall & (b) scaling behavior \\
\adjustbox{valign=t}{\includegraphics[width=0.45\linewidth,trim={0mm 0mm 0mm 7mm},clip]{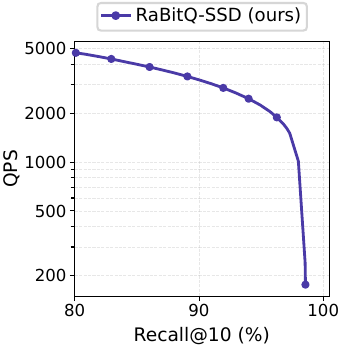}} &
\adjustbox{valign=t}{\scalebox{0.9}{
  \setlength{\tabcolsep}{4pt}
  \setlength{\tabcolsep}{3pt}%
\renewcommand{\arraystretch}{0.85}%
  \begin{tabular}{@{}lrrrr@{}}
    \toprule
     & \multicolumn{2}{c}{Recall $0.90$}
     & \multicolumn{2}{c}{Recall $0.95$} \\
    \cmidrule(lr){2-3}\cmidrule(l){4-5}
    Size & QPS & p/q & QPS & p/q \\
    \midrule
    $100$k & $13{,}501$ & $119.8$ & $10{,}978$ & $147.4$ \\
    $1$M   & $11{,}074$ & $147.7$ & $8{,}759$ & $186.6$ \\
    $10$M  & $7{,}875$  & $205.6$ & $5{,}895$ & $274.5$ \\
    $100$M & $5{,}795$  & $282.2$ & $4{,}197$ & $389.5$ \\
    $1$B   & $4{,}289$  & $381.2$ & $3{,}004$ & $544.2$ \\
    $10$B  & $3{,}210$  & $509.4$ & $2{,}226$ & $734.8$ \\
    \bottomrule
  \end{tabular}}}
  \\
  \end{tabular}
  \caption{
  \OURS on DINO-10B, $D_m=D/4=256$, $k=10$, and $T=64$.
  (a) throughput against recall
  (b) throughput and pages per query on subsets of the dataset.
  }
  \label{fig:tenb}
  \label{tab:dino-scalability}
\end{figure}

Figure~\ref{fig:tenb}(a) shows throughput against recall
on DINO-10B. \OURS achieves approximately
$3{,}210$ and $2{,}226$ QPS at recall $0.90$ and $0.95$,
respectively.
To examine how performance scales with dataset size,
we also evaluate subsets containing the first $100$k, $1$M, $10$M,
$100$M, and $1$B vectors,
still with $D_m=256$ and $C\approx N/4096$.
Figure~\ref{tab:dino-scalability}(b) and Figure~\ref{fig:dino-scaling-splash} summarize the results.
As the dataset grows from $100$k to $10$B points, a
$100{,}000$-fold increase in size, throughput decreases by just
of $4.21\times$ and $4.93\times$ at recalls $0.90$ and $0.95$, respectively.

With $D_m=D/4$, the shorter in-memory code prefix
loosens the lower bounds, allowing more candidates to survive
pruning and increasing SSD reads.
Across all six scales, the achieved read rates are about
$1.55$--$1.64$M pages per second, close to \rmax, indicating
I/O-bound search.
This is consistent with the tradeoff observed in the
varying-$D_m$ experiment (\S\ref{sec:exp:memdim}).

\subsection{Ablation Study}
\label{sec:exp:ablation}
The ablation varies three design choices: (1) which code is
held in memory, (2) which representation is stored on the SSD, and (3)
the precision used for re-ranking.
Each variant changes one component while keeping the other
index and system settings fixed.
Table~\ref{tab:ablation} reports QPS and pages read per query
at recall $0.90$ ($k=10$, $T=64$, $D_m=D$).
The results at other recalls are presented in Appendix~\ref{sec:app:ablation}.

\begin{table}[!htbp]
\caption{
Ablation of the in-memory code, on-disk representation, and re-ranking precision,
with recall $0.90$, $k=10$, $T=64$, $D_m=D$;
bits/dim = bits per dimension on disk, p/q = pages read per query. }
\label{tab:ablation}
\centering
\small
\setlength{\tabcolsep}{4pt}
\begin{tabular}{@{}lrrrrr@{}}
\toprule
& & \multicolumn{2}{c}{YFCC-100M} & \multicolumn{2}{c}{LAION-100M} \\
\cmidrule(lr){3-4}\cmidrule(l){5-6}
Variant & bits/dim & QPS & p/q & QPS & p/q \\
\midrule
\OURS ($B=9$) & 8 & \textbf{42,057} & 37.8 & \textbf{50,509} & 32.3 \\
\OURS-fp & 32 & 24,952 & 61.7 & 43,805 & 36.8 \\
\OURS (SQ8) & 8 & 41,659 & 37.5 & 50,051 & 32.6 \\
\OURS ($B=5$) & 4 & 40,688 & 38.3 & 48,099 & 33.2 \\
\OURS (SQ4) & 4 & 30,824 & 48.3 & 38,178 & 42.7 \\
\midrule
PipeANN (PQ) & 32 & 26,701 & 62.2 & 40,233 & 38.9 \\
PipeANN (RaBitQ) & 32 & 25,591 & 64.9 & 38,811 & 40.4 \\
\midrule
DiskANN  & 32 & 23,017 & 72.3 & 39,366 & 40.4 \\
DiskANN (disk-PQ) & 8 & 37,580 & 36.5 & 37,102 & 41.2 \\
\bottomrule
\end{tabular}%
\end{table}

\smallskip\noindent\textbf{The in-memory code.}
We evaluate several combinations of in-memory representations for the different methods. %
PipeANN supports both PQ and RaBitQ codes~\cite{pipeannrepo}, allowing to compare them directly.
Running it with $1$-bit RaBitQ codes
instead of PQ codes, with its disk graph and layout
unchanged, lowers throughput by $\sim4\%$ on both datasets:
merely replacing the in-memory code for navigation
offers no throughput benefit.
The superior results of \OURS come from the bound provided by Split-RaBitQ (\S\ref{sec:split:lb})
to prune candidates. %

\smallskip\noindent\textbf{The on-disk representation.}
\emph{\OURS-fp} stores the raw fp32 vector in place of our ex-codes, to match
the default storage of the baselines,
whereas \emph{DiskANN (disk-PQ)} replaces the
full-precision vector inside each graph node with a PQ code
of one byte per dimension.
On YFCC-100M, a compressed on-disk representation reduces the
pages needed per record from two to one, lowering page reads per query.
This improves throughput by $60$+\% for both DiskANN
and \OURS relative to their fp32 variants,
and \OURS still achieves $12\%$ higher
throughput than DiskANN (disk-PQ).

On LAION-100M, fp32 records already fit in one page, so
this effect does not apply: at recall $0.90$, DiskANN (disk-PQ) achieves $6\%$ lower
throughput than fp32 DiskANN.
This, in turn, stems from the quantization error introduced
by disk-PQ, which necessitates additional search work to reach
the same recall.
However, in \OURS, compact ex-codes stored contiguously
within inverted lists allow more surviving candidates to share a page read, which
yields $15\%$ higher
throughput over \OURS-fp.

\smallskip\noindent\textbf{Re-ranking precision.}
With SQ8 and SQ4, we replace ex-codes with $8$-bit or $4$-bit scalar
quantization with residual encoding~\cite{douze2026faisslib} on both datasets.
At recall $0.90$, $B=9$ and SQ8 have similar throughput,
whereas $B=5$ achieves $1.32\times$ and $1.26\times$ the throughput
of SQ4 on YFCC-100M and LAION-100M, respectively.
SQ4 requires more probing and page reads to reach the same recall.
This shows that \OURS combines the 1-bit
RaBitQ code and the 4-bit ex-code to produce a more accurate representation than just 4-bit SQ4.

\subsection{Additional baselines}
\label{sec:exp:spann}

\smallskip\noindent\textbf{SPANN.}
SPANN partitions the base into a large number of inverted lists,
keeps their centroids in memory, and reads the probed lists from
SSD~\cite{chen2021spann}.
Figure~\ref{fig:spann}(a)
reports the comparison
on Wiki-10M. %
Memory budgets are not matched: SPANN's
in-memory index takes $6.6$\,GB against our $1.1$\,GB.
At recall $0.90$, \OURS delivers $21.3\times$ SPANN's throughput, because
SPANN reads $870$ pages per query, compared to $\sim 38$ for
\OURS.
SPANN fetches whole selected posting lists, while
\OURS performs finer-grained SSD accesses.

\smallskip\noindent\textbf{AlayaLaser.}
AlayaLaser~\cite{chen2026alayalaser} is a graph-based index that accelerates distance
estimation using PCA and a SIMD-friendly layout that stores neighbors'
quantized codes in each node's on-disk record.
It caches node records in memory.
We use the graph parameters specified
in \S\ref{sec:exp:setup}, with beam width $16$ and PCA dimension $256$.
Therefore, for AlayaLaser, there is a throughput against memory usage tradeoff, which Figure~\ref{fig:alayalaser}(b) shows.
Under the memory budget matched to our
$D_m = D$ configuration, AlayaLaser caches about $2.8\%$ of the nodes.
AlayaLaser achieves $8{,}592$ QPS, compared
with $28{,}571$ QPS for \OURS. Limited cache coverage leaves
many node accesses to SSD, with each cache miss
reading a two-page node record. %
To reach \OURS's QPS performance, AlayaLaser needs to cache $\sim 50\%$ of the nodes (30~GB of memory),
at which point the relevance of a hybrid memory-SSD setting becomes questionable.

\begin{figure}[t]
  \centering
  \includegraphics[width=0.9\linewidth]{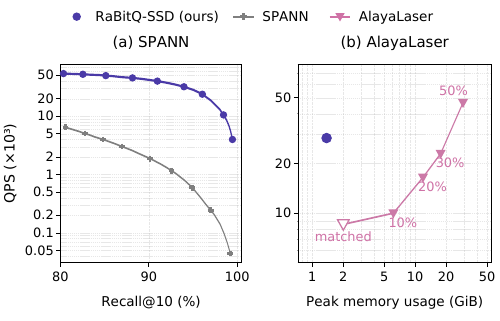}
  \caption{Comparisons on Wiki-10M ($T = 64$, $k = 10$;
  \OURS at $D_m = D$). (a) Throughput against recall; SPANN uses
  its default configuration. (b) Throughput against
  peak memory usage at recall
  $0.95$. ``matched'' denotes AlayaLaser's configuration under the
  matched memory budget; percentages denote cached nodes.}
  \label{fig:spann}
  \label{fig:alayalaser}
\end{figure}

\subsection{Indexing Cost and Footprints}
\label{sec:exp:cost}

Table~\ref{tab:cost} reports the end-to-end build time
and the on-SSD index size of \OURS and the three
main graph baselines on six datasets, at $D_m = D$.
All builds run on CPU only, with $64$ threads and the
build-memory budget of \S\ref{sec:exp:setup}.

\OURS builds indexes at least $5.0\times$ faster
than the fastest graph baseline on each dataset. On DataComp-1B,
it takes $9.2$ hours, compared with $45.6$ hours for DiskANN.
Our build has three stages: training the coarse quantizer,
assignment, and encoding. Training runs two-level $k$-means~\cite{faiss2level} on
a sample of the training vectors; assignment uses HNSW search
over the centroids rather than an exhaustive centroid scan for
each point. On DataComp-1B, the three stages take
$1.8$, $5.2$, and $2.2$ hours, respectively.

\OURS stores compact ex-codes and factors
on SSD, whereas the default graph indexes store
full-precision vectors and adjacency lists. Using compact codes
and omitting per-point adjacency lists reduces our on-SSD index
size by a factor of $4.0$--$7.0$
relative to these baselines. On DataComp-1B,
\OURS's on-SSD index occupies
$587$\,GB, compared with $4{,}096$\,GB
for DiskANN and PipeANN.
DiskANN's disk-PQ variant, which stores PQ codes on the SSD in
place of the vectors, narrows this storage gap to
$1.5$--$1.7\times$ on YFCC-100M and LAION-100M
(\S\ref{sec:exp:ablation}).

\begin{table}[!htbp]
\caption{End-to-end build time (hours) / on-SSD index
size (GB), at $D_m = D$. Results use each method's
default disk representation. ``--'': not built.}
\label{tab:cost}
\centering
\small
\setlength{\tabcolsep}{3pt}
\resizebox{\columnwidth}{!}{%
\begin{tabular}{@{}lrrrrrr@{}}
\toprule
dataset & OpenAI & Wiki & DPR & LAION & YFCC & DataComp \\
$N$ & 5M & 10M & 100M & 100M & 100M & 1B \\
\midrule
\textbf{\OURS}
& $\mathbf{0.05}/\mathbf{10.2}$
& $\mathbf{0.06}/\mathbf{8.2}$
& $\mathbf{0.9}/\mathbf{83.1}$
& $\mathbf{0.8}/\mathbf{83.8}$
& $\mathbf{1.5}/\mathbf{135}$
& $\mathbf{9.2}/\mathbf{587}$ \\
DiskANN
& $0.5/40.9$
& $0.5/41.0$
& $7.9/414$
& $5.5/418$
& $9.1/809$
& $45.6/4{,}096$ \\
PipeANN
& $0.4/40.9$
& $0.5/41.0$
& $8.8/414$
& $6.3/418$
& $9.3/809$
& $45.9/4{,}096$ \\
Starling
& $0.9/40.9$
& $1.0/41.0$
& $19.6/414$
& $12.1/418$
& $18.7/809$
& -- \\
\bottomrule
\end{tabular}%
}
\end{table}

%% file: content/6-related_work.tex
\section{Related Work}
\label{sec:related}

\noindent\textbf{SSD-resident vector search.}
Among graph-based methods,
DiskANN~\cite{subramanya2019diskann} sets the pattern most later
methods follow: compressed codes stay in memory to guide a best-first
traversal, while the graph itself resides on SSD.
Subsequent graph-based methods reorganize the on-disk layout so
that each fetched page serves more of the
search~\cite{wang2024starling,ni2023diskannpp,kang2025pageann,yue2025selectedges,yin2025gorgeous,bamg,govector}
and overlap the traversal with asynchronous
reads~\cite{guo2025pipeann,kang2026laann}.
AlayaLaser accelerates distance computation~\cite{chen2026alayalaser}.
A recent study explores this design space~\cite{li2026designspace}.
Among partitioning-based methods,
SPANN~\cite{chen2021spann} keeps centroids in memory and reads probed lists from SSD;
HELMSMAN~\cite{huang2026helmsman} rebuilds the I/O path on SPDK to
saturate its NVMe arrays and uses learned models to choose the
probed lists; and OrchANN~\cite{huan2025orchann} orchestrates out-of-core
search under skewed partitions.
Among other methods,
LM-DiskANN~\cite{lm_diskann} and AiSAQ~\cite{tatsuno2024aisaq}
move even the compressed codes to the SSD;
VeloANN~\cite{zhao2026veloann} stores RaBitQ ex-codes on the
SSD, with its in-memory $1$-bit code ordering the graph walk; and
SkipDisk~\cite{song2026skipdisk} prunes reads with a
triangle-inequality bound tightened by per-point pivots.
Furthermore, some disk-resident indexes have been extended to support
updates~\cite{singh2021freshdiskann,xu2023spfresh,guo2026odinann,song2026navis,zhong2025lsmvec}.
In \OURS, the in-memory prefix of the $1$-bit RaBitQ code gives every
point its own bound (\S\ref{sec:split}), and only the pages of
surviving candidates are read (\S\ref{sec:pipeline}).

\smallskip\noindent\textbf{Vector quantization.}
A substantial body of work has investigated the quantization of high-dimensional vectors across multiple fields, including machine learning, computer vision, and data management~\cite{jegou2011pq,additivePQ,ITQ,vaq,composite_quantization,ge2013optimized,lsq++,vaplusfile,splitvq,vafile,aguerrebere2023lvq,codesign2023,guo2020scann,vargaftik2021drive,eden,RaBitQ,ExtendedRaBitQ,turboquant}. We refer readers to existing surveys, experimental studies, and textbooks for more comprehensive overviews of related techniques~\cite{learningtohash,surveyl2hash,ite_matsui_2018_pq_survey,dataseries_benchmark,revision1_bound,revision2_bound,libertyvq}.
This paper builds on RaBitQ~\cite{RaBitQ}, which provides
probabilistic bounds on distance-estimation error, and its
extension supporting flexible bit budgets~\cite{ExtendedRaBitQ}
(\S\ref{sec:pre:rabitq}).
In particular, RaBitQ has been integrated with in-memory graph
indexes~\cite{symphonyqg2025}, made GPU-native~\cite{GPUIVFRaBitQ},
and run as a heterogeneous NPU--CPU
system~\cite{he2026ascendrabitq}.
Some recent quantization methods are built on top of RaBitQ:
MRQ~\cite{yang2026mrq} truncates PCA-projected dimensions and quantizes the kept dimensions with RaBitQ; and
SAQ~\cite{li2025saq} allocates bits across segments of PCA-projected dimensions non-uniformly and quantizes each segment with a variant of RaBitQ.
Split-RaBitQ instead derives an unbiased distance estimator and a probabilistic error
bound from a prefix of the 1-bit RaBitQ code.

%% file: content/7-conclusion.tex
\section{Conclusion}
\label{sec:conclusion}

We presented \OURS, a partitioning-based algorithm for SSD-resident vector search. Its \emph{Split-RaBitQ} design retains a configurable code prefix in memory, enabling fine-grained pruning under different memory budgets.
Its asynchronous pipeline overlaps computation with I/O
while reducing unnecessary SSD reads.
Experiments show that \OURS provides higher throughput, fewer SSD reads, smaller indexes, and faster construction than comparable graph-based systems from the state of the art. It also searches 10 billion vectors using a single server and one NVMe SSD, demonstrating its scalability.

%% file: content/appendix_proof.tex
\appendix
\section{Analysis of Split-RaBitQ}
\label{sec:app:proof}

\begin{figure}[!htbp]
  \centering
  \includegraphics[width=\linewidth]{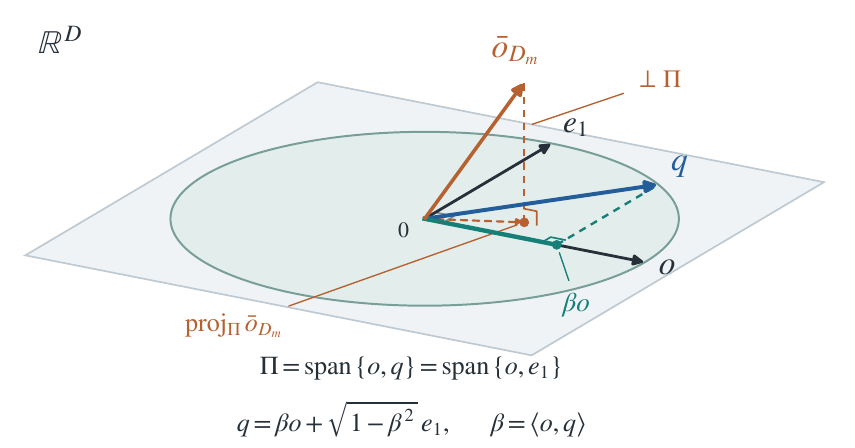}
  \caption{Geometry of Split-RaBitQ in $\mathbb{R}^{D}$.
  The shaded plane $\Pi=\operatorname{span}\{o,q\}$ contains the
  orthogonal unit vectors $o$ and $e_1$, and the unit query $q$;
  the ellipse depicts its unit circle.
  The prefix vector $\bar{o}_{D_m}$, of norm $\sqrt{D_m/D}$,
  is projected onto $\Pi$.
  Its component perpendicular to $\Pi$ contributes nothing to
  $\langle\bar{o}_{D_m},q\rangle$.}
  \label{fig:split-geometry}
\end{figure}

Throughout this appendix, all probabilities and expectations are over
the random orthogonal matrix $P$, which, as in~\cite{RaBitQTR}, is
generated by Gram--Schmidt orthonormalization of a matrix with
independent standard Gaussian entries; $P^{-1}o$ then follows the
uniform distribution on the unit sphere (Section~B.2
of~\cite{RaBitQTR}). The constants $c_0$ and $c$ are those in the
coordinate tail bound~(32) and the Lipschitz concentration
inequality~(41) of~\cite{RaBitQTR}, and, by the expression in
Lemma~B.3 of~\cite{RaBitQTR}, $\mathbb{E}\langle \bar{o}, o\rangle$
exceeds $0.79$ for every $D \ge 2$. Eq.~\eqref{eq:alignment}, i.e.,
(30) of~\cite{RaBitQTR}, applied to $\bar{x}$ and to $\bar{x}_{D_m}$
gives
\begin{equation}
\label{eq:app:l1}
\langle \bar{o}, o\rangle
= \frac{1}{\sqrt{D}}\sum_{i=1}^{D}\bigl\lvert (P^{-1}o)[i]\bigr\rvert ,
\qquad
\langle \bar{o}_{D_m}, o\rangle
= \frac{1}{\sqrt{D}}\sum_{i=1}^{D_m}\bigl\lvert (P^{-1}o)[i]\bigr\rvert ,
\end{equation}
since $\bar{x}[i] = \operatorname{sign}\bigl((P^{-1}o)[i]\bigr)/\sqrt{D}$
and $\bar{x}_{D_m}$ keeps the first $D_m$ of these entries. In
particular $\langle \bar{o}, o\rangle \ge 1/\sqrt{D}$, because the
$\ell_1$ norm of the unit vector $P^{-1}o$ is at least its $\ell_2$
norm.

\subsection{Joint Distribution of the Prefix's Components}

Lemma~B.3 of~\cite{RaBitQTR} states that the joint distribution of
$(\langle \bar{o}, o\rangle, \langle \bar{o}, e_1\rangle)$ is identical to
that of $(\langle \bar{o}, o\rangle, \sqrt{1 - \langle \bar{o}, o\rangle^2}\,X_1)$,
where $X_1$ follows the distribution $p_{D-1}$ of Lemma~B.1
in~\cite{RaBitQTR} and is independent of $\langle \bar{o}, o\rangle$.
The following lemma extends it to the prefix.

\begin{lemma}
\label{lem:app:joint}
Let $o$ and $e_1$ be two unit vectors with $o \perp e_1$, and let
$\bar{x}$, $\bar{o}$, $\bar{x}_{D_m}$, and $\bar{o}_{D_m}$ be as in
\S\ref{sec:split:twotier}. The joint distribution of
$(\langle \bar{o}, o\rangle,\, \langle \bar{o}_{D_m}, o\rangle,\,
  \langle \bar{o}_{D_m}, e_1\rangle)$ is identical to that of
\begin{equation}
\label{eq:app:joint}
\Bigl(\langle \bar{o}, o\rangle,\;
      \langle \bar{o}_{D_m}, o\rangle,\;
      \sqrt{D_m/D - \langle \bar{o}_{D_m}, o\rangle^2}\; X_1\Bigr) ,
\end{equation}
where $X_1$ follows the distribution $p_{D-1}$ of Lemma~B.1
in~\cite{RaBitQTR} and is independent of $P^{-1}o$, hence of the
first two components.
\end{lemma}

\begin{proof}
We follow the proof of Lemma~B.3 in~\cite{RaBitQTR} step by step,
with $\bar{x}_{D_m}$ in place of $\bar{x}$ where the two differ.

\emph{Reduction.} Let $u_1 := (1, 0, \dots, 0)$ and
$u_2 := (0, 1, 0, \dots, 0)$. As in~\cite{RaBitQTR}, since the
distribution of $UP$ is identical to that of $P$ for every orthogonal
matrix $U$, it suffices to prove the lemma for $o = u_1$ and
$e_1 = u_2$. Let $p_1$ and $p_2$ be the first two rows of $P$, as
in~\cite{RaBitQTR}; then $P^{-1}o = P^{\top}u_1 = p_1$ and
$P^{-1}e_1 = p_2$.

\emph{Deferred decision.} The rows of $P$ are generated one by one
((46) of~\cite{RaBitQTR}): $p_1$ is a normalized standard Gaussian
vector, and
$p_2 = (g - \langle g, p_1\rangle p_1)/\lVert g - \langle g, p_1\rangle p_1\rVert$
for a standard Gaussian vector $g$ independent of $p_1$. Once $p_1$
is sampled, $\bar{x} = \arg\max_{x \in \mathcal{C}}\langle x, p_1\rangle$
is fixed ((47)--(48) of~\cite{RaBitQTR}), and so is $\bar{x}_{D_m}$,
obtained from $\bar{x}$ by setting its coordinates beyond the first
$D_m$ to zero, an operation that does not involve $P$. Hence
$\langle \bar{o}, o\rangle = \langle \bar{x}, p_1\rangle$ and
$\langle \bar{o}_{D_m}, o\rangle = \langle \bar{x}_{D_m}, p_1\rangle$
are functions of $p_1$, whereas
$\langle \bar{o}_{D_m}, e_1\rangle
= \langle P\bar{x}_{D_m}, u_2\rangle
= \langle \bar{x}_{D_m}, p_2\rangle$ involves $p_2$ as well. It
remains to determine the conditional distribution of
$\langle \bar{x}_{D_m}, p_2\rangle$ given $p_1$.

\emph{Alignment.} Conditional on $p_1$, the vector $\bar{x}_{D_m}$ is
fixed, with $\lVert \bar{x}_{D_m}\rVert^2 = D_m/D$ (it has $D_m$
entries of magnitude $1/\sqrt{D}$). Let
\begin{equation*}
\begin{aligned}
v_1 &:= (1, 0, \dots, 0), \\
v_2 &:= \Bigl(\langle \bar{x}_{D_m}, p_1\rangle,\,
\sqrt{D_m/D - \langle \bar{x}_{D_m}, p_1\rangle^2},\, 0, \dots, 0\Bigr).
\end{aligned}
\end{equation*}
The pairs $(p_1, \bar{x}_{D_m})$ and $(v_1, v_2)$ have the same norms
and the same inner product, so there is an orthogonal matrix $V$ with
$Vp_1 = v_1$ and $V\bar{x}_{D_m} = v_2$. This is the alignment step
of~\cite{RaBitQTR}, where $\bar{x}$ has unit norm and the second entry
of $v_2$ is $\sqrt{1 - \langle \bar{x}, p_1\rangle^2}$; for the prefix
the norm is $\sqrt{D_m/D}$. Then, as in (49)--(51)
of~\cite{RaBitQTR},
\begin{align}
\langle \bar{x}_{D_m}, p_2\rangle
&= \Bigl\langle \bar{x}_{D_m},
   \frac{g - \langle g, p_1\rangle\, p_1}
        {\lVert g - \langle g, p_1\rangle\, p_1\rVert}\Bigr\rangle
\label{eq:app:align1}\\
&= \Bigl\langle V\bar{x}_{D_m},
   \frac{Vg - \langle Vg, Vp_1\rangle\, Vp_1}
        {\lVert Vg - \langle Vg, Vp_1\rangle\, Vp_1\rVert}\Bigr\rangle
\label{eq:app:align2}\\
&= \Bigl\langle v_2,
   \frac{g - \langle g, v_1\rangle\, v_1}
        {\lVert g - \langle g, v_1\rangle\, v_1\rVert}\Bigr\rangle ,
\label{eq:app:align3}
\end{align}
where \eqref{eq:app:align1} is the Gram--Schmidt step,
\eqref{eq:app:align2} holds because inner products and norms are
invariant under $V$, and \eqref{eq:app:align3} substitutes $Vp_1 = v_1$,
$V\bar{x}_{D_m} = v_2$, and $Vg$ by $g$, since the standard Gaussian
vector is rotation invariant and $g$ is independent of $p_1$. The
vector $g - \langle g, v_1\rangle v_1$ has first entry $0$ and its
remaining $D-1$ entries are independent standard Gaussians; after
normalization, the remaining $D-1$ entries follow the uniform
distribution on the unit sphere in the $(D-1)$-dimensional space and
are independent of $p_1$. Recall that
$v_2 = \bigl(\langle \bar{x}_{D_m}, p_1\rangle,
\sqrt{D_m/D - \langle \bar{x}_{D_m}, p_1\rangle^2}, 0, \dots, 0\bigr)$.
Thus, $\langle \bar{x}_{D_m}, p_2\rangle
= \sqrt{D_m/D - \langle \bar{x}_{D_m}, p_1\rangle^2}\cdot X_1$, where
$X_1$ follows the distribution $p_{D-1}$ and is independent of $p_1$.
Since $\langle \bar{x}_{D_m}, p_2\rangle = \langle \bar{o}_{D_m}, e_1\rangle$
and $\langle \bar{x}_{D_m}, p_1\rangle = \langle \bar{o}_{D_m}, o\rangle$,
this is the third component of Eq.~\eqref{eq:app:joint}, and the
independence of $X_1$ from $p_1$ gives its independence from the
first two components, which are functions of $p_1$.
\end{proof}

\subsection{Concentration of the Prefix's Share}

Lemma~B.2 of~\cite{RaBitQTR} states that for an $L$-Lipschitz
function $f$ on the unit sphere and a uniformly random unit vector $x$,
$\Pr\{\lvert f(x) - \mathbb{E} f(x)\rvert \ge t\}
\le 2\exp(-cDt^2/L^2)$ ((41) there), and applies it with
$f(x) = \lVert x\rVert_{1}/\sqrt{D}$ and $x = P^{-1}o$ to obtain, for
every $u > 0$,
\begin{equation}
\label{eq:app:B3}
\Pr\bigl\{\lvert \langle \bar{o}, o\rangle - \mathbb{E}\langle \bar{o}, o\rangle\rvert
\ge u/\sqrt{D}\bigr\} \le 2e^{-cu^2}
\end{equation}
((54) there). The same inequality controls the prefix's share.

\begin{lemma}
\label{lem:app:share}
For every $u > 0$,
\begin{equation*}
\Pr\Bigl\{\Bigl\lvert \langle \bar{o}_{D_m}, o\rangle
- \frac{D_m}{D}\,\mathbb{E}\langle \bar{o}, o\rangle\Bigr\rvert
\ge \frac{u\sqrt{D_m}}{D}\Bigr\} \le 2e^{-cu^2} .
\end{equation*}
\end{lemma}

\begin{proof}
By Eq.~\eqref{eq:app:l1}, $\langle \bar{o}_{D_m}, o\rangle = f(P^{-1}o)$
with $f(x) := \frac{1}{\sqrt{D}}\sum_{i=1}^{D_m}\lvert x[i]\rvert$,
and $P^{-1}o$ is a uniformly random unit vector, so its coordinates
are exchangeable and $\mathbb{E}\lvert (P^{-1}o)[i]\rvert$ does not
depend on $i$. By Eq.~\eqref{eq:app:l1},
\begin{equation*}
\begin{aligned}
\mathbb{E} f(P^{-1}o)
&= \frac{D_m}{\sqrt{D}}\,\mathbb{E}\bigl\lvert (P^{-1}o)[1]\bigr\rvert \\
&= \frac{D_m}{D}\cdot
  \frac{1}{\sqrt{D}}\sum_{i=1}^{D}\mathbb{E}\bigl\lvert (P^{-1}o)[i]\bigr\rvert \\
&= \frac{D_m}{D}\,\mathbb{E}\langle \bar{o}, o\rangle .
\end{aligned}
\end{equation*}
For unit vectors $x$ and $y$,
$\lvert f(x) - f(y)\rvert
\le \frac{1}{\sqrt{D}}\sum_{i=1}^{D_m}\lvert x[i] - y[i]\rvert
\le \sqrt{D_m/D}\,\lVert x - y\rVert$
by the triangle inequality and the Cauchy--Schwarz inequality, so $f$
is $\sqrt{D_m/D}$-Lipschitz on the unit sphere. Lemma~B.2
of~\cite{RaBitQTR} with $L = \sqrt{D_m/D}$ and $t = u\sqrt{D_m}/D$
gives the bound, since $cDt^2/L^2 = cu^2$.
\end{proof}

\subsection{Proof of Theorem~\ref{thm:split-lb}}

Recall Eq.~\eqref{eq:split-decomp}: with $\beta = \langle o, q\rangle$
and $A_{D_m} = (D/D_m)\,\langle \bar{o}_{D_m}, o\rangle / \langle \bar{o}, o\rangle$,
the error of the estimator is
$\beta\,(A_{D_m} - 1)
+ \sqrt{1-\beta^2}\,(D/D_m)\,\langle \bar{o}_{D_m}, e_1\rangle / \langle \bar{o}, o\rangle$.
When $o$ and $q$ are collinear, only the first term is present.

\smallskip\noindent\emph{Unbiasedness.}
We show that both terms have mean zero. For the first term, by
Eq.~\eqref{eq:app:l1},
\begin{equation}
\label{eq:app:A}
A_{D_m}
= \frac{D}{D_m}\cdot
  \frac{\sum_{i=1}^{D_m}\lvert (P^{-1}o)[i]\rvert}
       {\sum_{j=1}^{D}\lvert (P^{-1}o)[j]\rvert} .
\end{equation}
The coordinates of the uniformly random unit vector $P^{-1}o$ are
exchangeable, hence so are the ratios
$\lvert (P^{-1}o)[i]\rvert / \sum_{j}\lvert (P^{-1}o)[j]\rvert$,
$i = 1, \dots, D$, which sum to $1$; each ratio therefore has
expectation $1/D$, the sum of the first $D_m$ of them has expectation
$D_m/D$, and $\mathbb{E}[A_{D_m}] = 1$. (This uses
that the prefix is a fixed set of coordinates, chosen without
reference to $P$.) For the second term, Lemma~\ref{lem:app:joint}
gives
\begin{align}
\mathbb{E}\biggl[\frac{D}{D_m}
\frac{\langle \bar{o}_{D_m}, e_1\rangle}{\langle \bar{o}, o\rangle}\biggr]
&= \mathbb{E}\biggl[\frac{D}{D_m}
\frac{\sqrt{D_m/D - \langle \bar{o}_{D_m}, o\rangle^2}}
     {\langle \bar{o}, o\rangle}\, X_1\biggr]
\label{eq:app:unb1}\\
&= \mathbb{E}\biggl[\frac{D}{D_m}
\frac{\sqrt{D_m/D - \langle \bar{o}_{D_m}, o\rangle^2}}
     {\langle \bar{o}, o\rangle}\biggr]
\cdot \mathbb{E}[X_1]
\label{eq:app:unb2}\\
&= 0 ,
\label{eq:app:unb3}
\end{align}
where \eqref{eq:app:unb1} is by Lemma~\ref{lem:app:joint},
\eqref{eq:app:unb2} is due to the independence between $X_1$ and
$P^{-1}o$ (the first factor is a function of $P^{-1}o$, bounded
because $\langle \bar{o}, o\rangle \ge 1/\sqrt{D}$), and
\eqref{eq:app:unb3} is because $p_{D-1}$ has mean $0$. Multiplied by
$\sqrt{1-\beta^2}$, the second term has mean zero, which completes
the proof of the unbiasedness.

\smallskip\noindent\emph{Error bound.}
By Lemma~\ref{lem:app:joint}, the error has the same distribution as
\begin{equation}
\label{eq:app:Z}
Z := \beta\,(A_{D_m} - 1)
+ \sqrt{1-\beta^2}\;\frac{D}{D_m}
  \frac{\sqrt{D_m/D - \langle \bar{o}_{D_m}, o\rangle^2}}
       {\langle \bar{o}, o\rangle}\, X_1 ,
\end{equation}
so it suffices to bound $Z$. Fix $0 < \epsilon \le \sqrt{D_m}/3$ and
consider three events:
\begin{itemize}
\item $\mathcal{E}_1$: $\bigl\lvert \langle \bar{o}_{D_m}, o\rangle
  - (D_m/D)\,\mathbb{E}\langle \bar{o}, o\rangle\bigr\rvert
  \le \epsilon\sqrt{D_m}/D$. By Lemma~\ref{lem:app:share} with
  $u = \epsilon$, $\Pr\{\mathcal{E}_1^{\mathrm{c}}\} \le 2e^{-c\epsilon^2}$.
\item $\mathcal{E}_2$: $\bigl\lvert \langle \bar{o}, o\rangle
  - \mathbb{E}\langle \bar{o}, o\rangle\bigr\rvert
  \le \epsilon/\sqrt{D}$. By
  Eq.~\eqref{eq:app:B3} with $u = \epsilon$,
  $\Pr\{\mathcal{E}_2^{\mathrm{c}}\} \le 2e^{-c\epsilon^2}$. On
  $\mathcal{E}_2$, $\langle \bar{o}, o\rangle \ge \mathbb{E}\langle \bar{o}, o\rangle/2$,
  because $\epsilon/\sqrt{D} \le \epsilon/\sqrt{D_m} \le 1/3
  < 0.79/2 < \mathbb{E}\langle \bar{o}, o\rangle/2$.
\item $\mathcal{E}_3$: $\lvert X_1\rvert \le \epsilon/\sqrt{D-1}$.
  By the tail bound~(32) of~\cite{RaBitQTR} in dimension $D-1$,
  $\Pr\{\mathcal{E}_3^{\mathrm{c}}\} \le 2e^{-c_0\epsilon^2}$.
\end{itemize}
On $\mathcal{E}_1 \cap \mathcal{E}_2$, by Eq.~\eqref{eq:app:l1} and
the triangle inequality,
\begin{equation*}
\begin{aligned}
\lvert A_{D_m} - 1\rvert
&= \frac{\bigl\lvert (D/D_m)\,\langle \bar{o}_{D_m}, o\rangle
        - \langle \bar{o}, o\rangle\bigr\rvert}
       {\langle \bar{o}, o\rangle} \\
&\le \frac{(D/D_m)\cdot\epsilon\sqrt{D_m}/D + \epsilon/\sqrt{D}}
         {\mathbb{E}\langle \bar{o}, o\rangle/2}
\le \frac{4}{\mathbb{E}\langle \bar{o}, o\rangle}\cdot
    \frac{\epsilon}{\sqrt{D_m}} ,
\end{aligned}
\end{equation*}
where the numerator bounds
$\bigl\lvert(D/D_m)\,\langle \bar{o}_{D_m}, o\rangle - \mathbb{E}\langle \bar{o}, o\rangle\bigr\rvert$
by $\mathcal{E}_1$ and
$\bigl\lvert\mathbb{E}\langle \bar{o}, o\rangle - \langle \bar{o}, o\rangle\bigr\rvert$ by
$\mathcal{E}_2$, and $\epsilon/\sqrt{D} \le \epsilon/\sqrt{D_m}$ gives the
last step.

On $\mathcal{E}_2 \cap \mathcal{E}_3$, the second term of $Z$ is at most
(since $\sqrt{D_m/D - \langle \bar{o}_{D_m}, o\rangle^2} \le \sqrt{D_m/D}$
and $\sqrt{1-\beta^2} \le 1$)
\begin{equation*}
\begin{aligned}
\frac{D}{D_m}\sqrt{\frac{D_m}{D}}\cdot
\frac{\epsilon/\sqrt{D-1}}{\mathbb{E}\langle \bar{o}, o\rangle/2}
&= \frac{2}{\mathbb{E}\langle \bar{o}, o\rangle}
  \sqrt{\frac{D}{D-1}}\cdot\frac{\epsilon}{\sqrt{D_m}} \\
&\le \frac{2\sqrt{2}}{\mathbb{E}\langle \bar{o}, o\rangle}\cdot
    \frac{\epsilon}{\sqrt{D_m}}
\end{aligned}
\end{equation*}
in magnitude, using $D/(D-1) \le 2$. Hence, with $\lvert\beta\rvert \le 1$,
on $\mathcal{E}_1 \cap \mathcal{E}_2 \cap \mathcal{E}_3$,
\begin{equation*}
\lvert Z\rvert
\le \frac{4 + 2\sqrt{2}}{\mathbb{E}\langle \bar{o}, o\rangle}\cdot
    \frac{\epsilon}{\sqrt{D_m}}
< \frac{9\,\epsilon}{\sqrt{D_m}} ,
\end{equation*}
the last step because $\mathbb{E}\langle \bar{o}, o\rangle > 0.79$ and
$(4 + 2\sqrt{2})/0.79 < 9$. By the union bound,
\begin{equation}
\label{eq:app:union}
\Pr\bigl\{\lvert Z\rvert > 9\,\epsilon/\sqrt{D_m}\bigr\}
\le \sum_{j=1}^{3}\Pr\bigl\{\mathcal{E}_j^{\mathrm{c}}\bigr\}
\le 4e^{-c\epsilon^2} + 2e^{-c_0\epsilon^2}
\le 6e^{-c'\epsilon^2} ,
\end{equation}
where $c' := \min(c_0, c)$. Since a probability is at most $1$, the
right-hand side of Eq.~\eqref{eq:app:union} can be replaced by
$2e^{-c''\epsilon^2}$ with $c'' := c'\ln 2/\ln 6$: when
$c'\epsilon^2 \le \ln 6$ the new bound is at least $1$, and otherwise
$6e^{-c'\epsilon^2} \le 2e^{-c''\epsilon^2}$ because
$(c' - c'')\epsilon^2 = (c'\ln 3/\ln 6)\,\epsilon^2 > \ln 3$. Finally,
applying the result with $\epsilon/9$ in place of $\epsilon$ (which
stays within the range $\epsilon/9 \le \sqrt{D_m}/3$) gives
$\Pr\{\lvert Z\rvert > \epsilon/\sqrt{D_m}\} \le 2e^{-c_1\epsilon^2}$
with $c_1 := c''/81$, which is the error bound of
Theorem~\ref{thm:split-lb}.

\subsection{Empirical Error Distribution}
\label{sec:app:distribution}

We examine the estimation error for synthetic
unit-vector pairs in $\mathbb{R}^{768}$ with prescribed
inner products $\beta=\langle o,q\rangle\in\{0.5,0.9\}$.
For each fixed pair, we evaluate both $D_m=D$ and $D_m=D/2$
over $10^6$ independent Haar-distributed rotations, which
preserve $\beta$ while varying the quantization error.
Figure~\ref{fig:split-error-distribution} shows error distributions
centered near zero, consistent with the estimator's unbiasedness.
Retaining half the coordinates broadens both distributions,
reflecting the additional uncertainty introduced by the prefix.

\begin{figure}[!htbp]
  \centering
  \includegraphics[width=0.8\columnwidth]{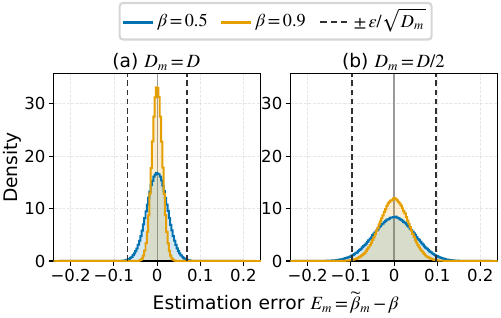}
  \caption{Empirical estimation-error distributions for
  $D=768$ with (a) $D_m=D$ and (b) $D_m=D/2$.
  Dashed lines mark $\pm\epsilon/\sqrt{D_m}$ for $\epsilon=1.9$.}
  \label{fig:split-error-distribution}
\end{figure}

\begin{table*}[!htbp]
  \centering
  \caption{Search-result determinism on YFCC-100M
  ($T=64$). Percentages give the fraction of queries with
  unchanged results across all eight runs, treating substitutions
  between points at the same true distance to the query as
  unchanged. The rank ranges refer to positions in the returned
  list, with $r=0.4k$ for both variants; order within each set
  is ignored. Recall spread is the difference between the
  maximum and minimum recall across runs, in percentage
  points (pp).}
  \label{tab:app:determinism}
  \small
  \setlength{\tabcolsep}{7pt}
  \begin{tabular}{@{}rlrrrr@{}}
    \toprule
    & & \multicolumn{3}{c}{Queries with unchanged results (\%)}
    & \\
    \cmidrule(lr){3-5}
    $k$ & Variant & Top-$k$ & $[1,r]$ & $(r,k]$
        & Recall spread (pp) \\
    \midrule
    10  & Scheduled & 99.95 & 100.00 & 99.95 & 0.003 \\
        & Full      & 97.89 & 99.98  & 97.89 & 0.022 \\
    \midrule
    100 & Scheduled & 99.96 & 100.00 & 99.96 & 0.0002 \\
        & Full      & 35.85 & 99.66  & 35.85 & 0.645 \\
    \bottomrule
  \end{tabular}%

\end{table*}

We use $\epsilon=1.9$, following RaBitQ~\cite{RaBitQ}, and count
a lower-bound violation when
$\widehat{\beta}_{D_m}-\beta<-\epsilon/\sqrt{D_m}$.
For $\beta=0.5$, the violation rate increases from $0.18\%$
at $D_m=D$ to $2.17\%$ at $D_m=D/2$.
The errors at $\beta=0.9$ are more concentrated:
no violations are observed with the full code in $10^6$ trials,
while the half-length prefix yields a rate of $0.21\%$.

\section{Search-Result Determinism}
\label{sec:app:determinism}

We examine search-result determinism using the $\nprobe$
selected for Full to achieve approximately $0.90$ recall, keeping
the same value for Scheduled. We execute the query workload eight
times, three in the original order and five with shuffled query
orders. Table~\ref{tab:app:determinism} compares the returned
neighbor sets across these runs.

Scheduled is nearly deterministic after accounting
for distance ties. Among $10{,}000$ queries, only $5$ at $k=10$
and $4$ at $k=100$ return different top-$k$ results across
the eight runs.
The argument of search-result determinism assumes exact re-ranking and
valid lower bounds.
The remaining changes occur near the rank-$k$ boundary,
where the affected candidates have lower bounds above their
true distances and the final $k$-th distances in runs that
omit them.
Such violations reflect the probabilistic nature of the lower
bound, also illustrated by the empirical error distributions
in \S\ref{sec:app:distribution}.
Although a candidate's lower bound is fixed across runs,
the candidate may be re-ranked before the running threshold
falls below that bound, or pruned if tested later.
Recall varies by at most $0.003$ percentage points across runs.

Full's output varies mainly in the tail. At $k=100$,
the unchanged-query fraction is $35.85\%$ for both
the entire top-$k$ and ranks $(0.4k,k]$, compared with
$99.66\%$ for the first $40$ positions.
This pattern is consistent with the $\alpha$-quantile
threshold making retention of neighbors at ranks
$(\alpha k,k]$ sensitive to I/O completion order.
Earlier re-ranking of closer candidates can tighten
the threshold before these neighbors' pages are issued,
whereas a different completion order may allow those
pages to be read.
Although I/O completion order affects the returned neighbors,
recall remains relatively stable across the eight runs,
ranging from $0.9005$ to $0.9069$ at $k=100$.

\section{Additional Evaluation Results}
\label{sec:app:additional}

\subsection{Main Evaluation}
\label{sec:app:main}

Table~\ref{tab:main-qps} reports throughput at recall
$0.90$ and $0.95$ for all five datasets in the main evaluation
(\S\ref{sec:exp:main}).

\begin{table*}[!htbp]
\caption{Throughput (QPS) at recall $0.90$ and $0.95$
on the five datasets in Figure~\ref{fig:main}, at $k=10$ and
$T=64$, under both memory budgets. Bold indicates the highest
QPS for each dataset, memory budget, and recall.}
\label{tab:main-qps}
\centering
\small
\setlength{\tabcolsep}{3pt}
\begin{tabular}{@{}clrrrrrrrr@{}}
\toprule
 & & \multicolumn{4}{c}{Recall $0.90$}
 & \multicolumn{4}{c}{Recall $0.95$} \\
\cmidrule(lr){3-6}\cmidrule(l){7-10}
Budget & Dataset & \OURS & DiskANN & PipeANN & Starling
 & \OURS & DiskANN & PipeANN & Starling \\
\midrule
\multirow{5}{*}{$D_m = D$} & OpenAI-5M & \textbf{43\,418} & 25\,759 & 26\,484 & 25\,466 & \textbf{22\,931} & 18\,717 & 18\,360 & 18\,598 \\
 & Wiki-10M & \textbf{42\,279} & 36\,847 & 31\,455 & 35\,965 & \textbf{28\,571} & 25\,209 & 20\,778 & 24\,181 \\
 & DPR-100M & \textbf{6\,690} & 2\,946 & 3\,834 & 3\,333 & \textbf{2\,829} & 1\,291 & 1\,736 & 1\,512 \\
 & LAION-100M & \textbf{50\,509} & 39\,366 & 40\,233 & 36\,109 & \textbf{32\,252} & 27\,175 & 27\,207 & 25\,740 \\
 & YFCC-100M & \textbf{42\,057} & 23\,017 & 26\,701 & 25\,568 & \textbf{22\,418} & 17\,788 & 19\,231 & 20\,073 \\
\midrule
\multirow{5}{*}{$D_m = D/2$} & OpenAI-5M & \textbf{32\,818} & 16\,695 & 16\,561 & 16\,590 & \textbf{24\,345} & 10\,234 & 9\,700 & 10\,147 \\
 & Wiki-10M & \textbf{21\,047} & 18\,559 & 16\,206 & 17\,283 & \textbf{17\,022} & 10\,388 & 9\,164 & 9\,827 \\
 & DPR-100M & \textbf{8\,298} & 310 & 592 & 338 & \textbf{3\,814} & 142 & 236 & 97 \\
 & LAION-100M & \textbf{26\,631} & 18\,976 & 19\,162 & 19\,177 & \textbf{21\,002} & 9\,383 & 9\,182 & 9\,858 \\
 & YFCC-100M & \textbf{36\,465} & 14\,851 & 16\,360 & 16\,367 & \textbf{26\,120} & 9\,427 & 9\,961 & 10\,349 \\
\bottomrule
\end{tabular}%

\end{table*}

\subsection{Ablation Study}
\label{sec:app:ablation}
Figure~\ref{fig:app:ablation} shows the throughput--recall
curves for the three ablations on YFCC-100M and LAION-100M.

On YFCC-100M, \OURS ($B=9$) is faster than DiskANN (disk-PQ)
at recall $0.90$, but the ordering reverses at recall $0.95$
(Figure~\ref{fig:app:ablation}(b)).
Reaching the higher recall requires more probing, increasing the
work spent scanning full-dimensional in-memory codes.
FastScan becomes the main bottleneck for \OURS, limiting the
throughput benefit of fewer SSD reads.

The gap between $B=5$ and SQ4 grows at higher recall on both
datasets (Figure~\ref{fig:app:ablation}(c,f)).
SQ4 requires substantially more probing to reach higher recall,
reducing throughput. In the measured sweep, SQ4 reaches recall
$0.95$ on YFCC-100M at much lower throughput than $B=5$, and
does not reach recall $0.95$ on LAION-100M.

\begin{figure*}[!htbp]
  \centering
  \includegraphics[width=0.8\textwidth]{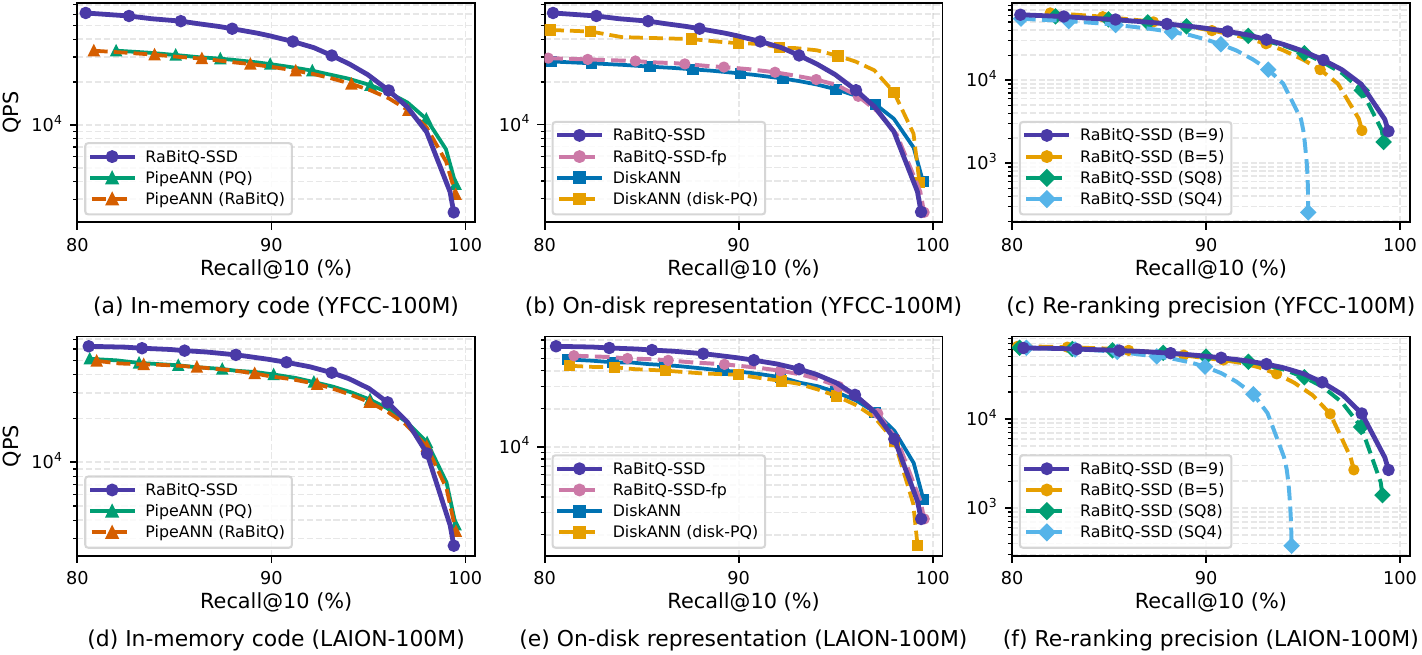}
  \caption{Ablation on YFCC-100M (top row) and LAION-100M
  (bottom row), at $D_m=D$, $k=10$, and $T=64$.
  The columns compare the in-memory code, on-disk representation,
  and re-ranking precision, respectively. All QPS axes use a log scale.}
  \label{fig:app:ablation}
\end{figure*}

\subsection{Parameter Study}
\label{sec:exp:param}

\noindent\textbf{The in-flight budget $Q$.}
Figure~\ref{fig:qdstudy} reports throughput at recall $0.90$
against $Q$ on YFCC-100M, one line per thread count
($D_m = D$). When $Q$ is too small, a thread has too few reads in flight to use
the SSD's internal parallelism. At $k = 1000$, where a query issues
hundreds of reads, throughput at $T = 1$ falls to about a
third as $Q$ goes from $16$ to $4$; at $k = 10$ the same change
barely changes throughput, because the few reads hide behind
compute, as the $T = 1$ timeline of Fig.~\ref{fig:timeline} showed. When $Q$ is
too large, the pops reach deeper into the priority set and commit
candidates that a slightly later threshold would have pruned
(\S\ref{sec:pipe:invlist}). Each such commit reads a page the query did not
need: the page is fetched, swept, and re-screened, some of its
members are re-ranked, and none of that work would have happened at
a smaller $Q$. Under a heavy I/O load the excess reads
also lengthen completion latency (\S\ref{sec:pipe:alpha}), so
candidates are committed against a staler threshold and fewer of
them are pruned. Past $Q = 16$ throughput declines slightly at
$T = 1$; at $T = 64$ and $k = 10$ it falls to about a third of its
value at the best $Q$. The best $Q$ therefore rises with $k$ and falls with $T$, and the deployed rule ($Q = 16$ for
$T \le 16$, $Q = 8$ above) stays within $3.7\%$ of the throughput at
the best $Q$ at every $T$ and $k$ measured.

\begin{figure}[!htbp]
  \centering
  \includegraphics[width=\linewidth]{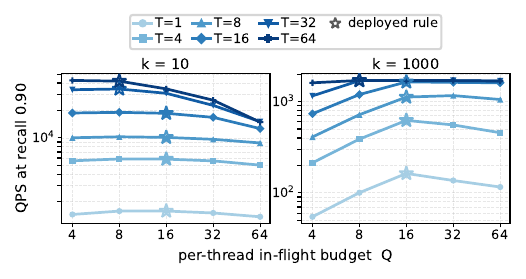}
  \caption{The in-flight budget $Q$ on YFCC-100M ($D_m = D$): QPS at
  recall $0.90$ against the per-thread $Q$, one line per thread count.
  Stars mark the deployed rule ($Q = 16$ for $T \le 16$, $8$ for
  $T > 16$). }
  \label{fig:qdstudy}
\end{figure}

\begin{figure}[!htbp]
  \centering
  \includegraphics[width=\linewidth]{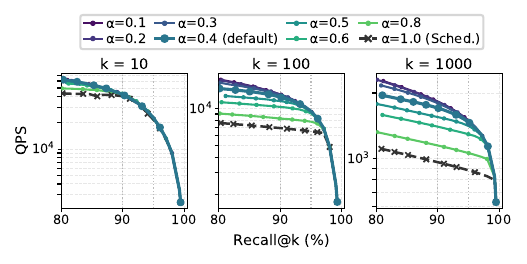}
  \caption{The drain threshold $\alpha$ on YFCC-100M at $T = 64$ and
  $D_m = D$: throughput against recall, one panel per $k$, one curve
  per $\alpha$ (log scale). $\alpha = 1.0$ disables the drain
  tightening and recovers Scheduled (\S\ref{sec:pipe:alpha}); the
  heavier curve is the default $\alpha = 0.4$.}
  \label{fig:alphastudy}
\end{figure}

\smallskip\noindent\textbf{The drain threshold $\alpha$.}
Figure~\ref{fig:alphastudy} reports throughput against recall for
$\alpha \in \{0.1, 0.2, 0.3, 0.4, 0.5, 0.6, 0.8, 1.0\}$ on
YFCC-100M at $T = 64$ and $D_m = D$. $\alpha = 1$ recovers Scheduled
(\S\ref{sec:pipe:alpha}); $\alpha = 0.4$ is the default. A tighter $\alpha$
prunes more of the pages left pending when the scan ends, so a
query reads fewer of them: at recall $0.90$ and $k = 100$, the
pages per query are $125$ at $\alpha = 0.1$, $132$ at the default,
and $236$ at $\alpha = 1$; at $k = 1000$ they are $893$, $957$, and
$1{,}699$.

Below the default the count barely moves, while above it nearly
doubles. At recall $0.90$ the default stays within $6.6\%$ of the
throughput at the best $\alpha$ at every $k$ measured.